\documentclass[aps,pra,showpacs,twoside,twocolumn,10pt]{revtex4-1}
\usepackage[colorlinks=true, citecolor=blue, urlcolor=blue ]{hyperref}
\usepackage{times,epsfig,amssymb,amsfonts,amsmath,bm,subfigure,mathtools,amsthm,braket,soul,enumitem,xcolor,physics,graphics,graphicx,tikz}
\UseRawInputEncoding
\usepackage[normalem]{ulem}
\usepackage{comment}
\usepackage{epstopdf}
\usepackage{relsize}
\newtheorem{theorem}{Theorem}
\newtheorem{corollary}[theorem]{Corollary}

\newtheorem{proposition}[theorem]{Proposition}

\DeclareMathOperator{\sur}{\text{sur}}

\begin{document}

\title{Deciding factor for detecting a particle within a subspace via  dark and bright  states}

\author{Aashay Pandharpatte$^{1,2}$, Pritam Halder$^2$, Aditi Sen(De)$^2$}
\affiliation{$^1$1
Indian Institute of Science Education and Research, Pune 411 008, India}
\affiliation{$^2$Harish-Chandra Research Institute, A CI of Homi Bhabha National Institute,  Chhatnag Road, Jhunsi, Allahabad - 211019, India}

\begin{abstract}
   In a measurement-induced continuous-time quantum walk, we address the problem of detecting a particle in a subspace, instead of a fixed position. In this configuration, we develop an approach of bright and dark states based on the unit and vanishing detection probability respectively for a particle-detection in the subspace.  Specifically, by employing the rank-nullity theorem, we determine several properties of dark and bright states in terms of  energy spectrum of the Hamiltonian used for a quantum walk and the projectors applied to detect the subspace. We provide certain conditions on the position and the rank of the subspace to be detected, resulting in the unit total detection probability, which has broad implications for quantum computing.  Further, we illustrate the forms of dark as well as bright states and  the dependence of detection probability on the number of dark states by considering a cyclic graph with nearest-neighbor and next nearest-neighbor hopping. Moreover, we observe that the divergence in the average number of  measurements for detecting a particle successfully in a subspace can be reduced by performing high rank projectors.

\end{abstract}
\maketitle

\section{Introduction}

The quantum mechanical analogue of a classical random walk, referred to as quantum walk \cite{kempe2003rev,andraca2012,Kadian2021}, can be classified into two distinct categories -- discrete-time and continuous-time quantum walk. Due to the quantum superposition principle, quantum walk represents a sophisticated framework for constructing quantum algorithms which, in turn, results in a universal paradigm for quantum computation \cite{Karafyllidis2005,childs2009,lovett2010,Kendon2014,Singh2021}. In particular, it has been utilized in a wide range of quantum information processing tasks, including   quantum search \cite{shenvi2003,Wong2015,Li2020}, quantum encryption and security \cite{Rohde2012,el-latif2020}, cryptographic systems \cite{Abd-El-Atty2021}, random number generation \cite{bae2021,bae2022}, state engineering \cite{vieira2013, innocenti2017, giordani2019,Kadian2021} to name a few. Additionally, quantum walks have been experimentally implemented \cite{Manouchehri2014} using   nuclear magnetic resonance \cite{du2003,ryan2005}, photonic  \cite{perets2008,bian2017} and optomechanical systems \cite{Moqadam2015}, and trapped ions \cite{Karski2009}.

One of the primary objectives of continuous-time quantum walks is to determine the probability and time of arriving at a certain location when a particle starts from a particular initial position. Despite controversies surrounding the consideration of time as an operator \cite{mielnik2011time}, significant progress have been achieved when addressing the time-of-arrival problem in the literature \cite{Allcock1969,Kijowski1974,Halliwell2009,Anastopoulos2006,Aharonov1998,Kumar1985,Galapon2012,Galapon2004,Galapon2005,Chakraborty2016}. Concurrently, several quantum search setups \cite{Grover1997,Childs2004,Magniez2011,Novo2015,Li2017} have been proposed, in conjunction with investigations into state transfer phenomena \cite{tanner2009, bohm2015}. 
In addition to the approaches, a periodic measurement strategy combined with unitary evolution -- which is determined by the Hamiltonian of a certain system -- can be employed to identify the particle.
 Within this realm,  the measurement process dynamically influences the evolution of the state of the walker.

 In stroboscopic measurement-induced quantum walk (MIQW) \cite{Didi2022}, the first-detected arrival problem becomes relevant \cite{Krovi2006hitting,krovi2006hypercube,Krovi2007quotient} over the first arrival time, which excludes stroboscopic measurements. Specifically, determining the walker in the target state using periodic measurements for the first time after beginning from some initial state is known as {\it the first detection problem}.  While measurements impede the quantum walker's free evolution, this problem has attracted a lot of  attention   \cite{Sinkovicz2016,Sinkovicz2015,Dhar2015,Grunbaum-schur-func2014,Caruso2009} since it is related to readout techniques in quantum computing tasks and   control of quantum systems \cite{butkovskiy1990control,Huang1983,Pierce1988,rabitz2000,Wiseman_Milburn_2009,campo2015}. Moreover, MIQW is intricately linked to mid-circuit measurements \cite{Govia2023,norcia2023,norcia2023yb}, a key component in quantum computing,  error correction \cite{hashim2024quasiprobabilistic} and mitigation \cite{botelho2022} and  has also been implemented on IBM quantum computers equipped with a mid-circuit readout feature \cite{wang2024hitting}. From a more fundamental perspective, this method of detecting a particle in a fixed position may further highlight the role of measurements in quantum theory 
\cite{Li2019,muller2022,feng2023,poboiko2024}.

The total probability of the first detection, referred to as  the \textit{total detection probability} \cite{Yamasaki2003},  is the statistics of the walker's detection for the first time during the application of an infinite number of stroboscopic measurements in MIQW \cite{Krovi2006hitting,krovi2006hypercube,Krovi2007quotient,FriedmanKessler2017,Allcock1969,Galapon2004,Galapon2005,Kumar1985,Aharonov1998,Anastopoulos2006,Allcock1969}. 
To emphasize, there exist certain initial conditions under which the desired state is never achieved due to destructive interference in the system; we refer to these states as dark states \cite{plenio1998,Barkai2020DarkState} in accordance with forbidden transitions in atomic physics and dark modes in quantum networks \cite{Chunhua2012,Kuzyk2017,Huang2022,huang2023darkmodetheoremsquantumnetworks}, whereas bright states can be recognized with certainty. It has been shown that the probability of identifying these states is connected with the existence of dark and bright energy states in the system \cite{Barkai2020DarkState,Liu2022}.   Recently, the method of calculating the number of dark modes in the context of opto-mechanical quantum networks has also  been proposed  \cite{huang2023darkmodetheoremsquantumnetworks}. However, in the framework of MIQW,  no methods for  computing the dark states in the case of subspace detection exist in the literature which is one of the main goals of this work.

\begin{figure}[h]
        \centering
        \includegraphics[width=\linewidth]{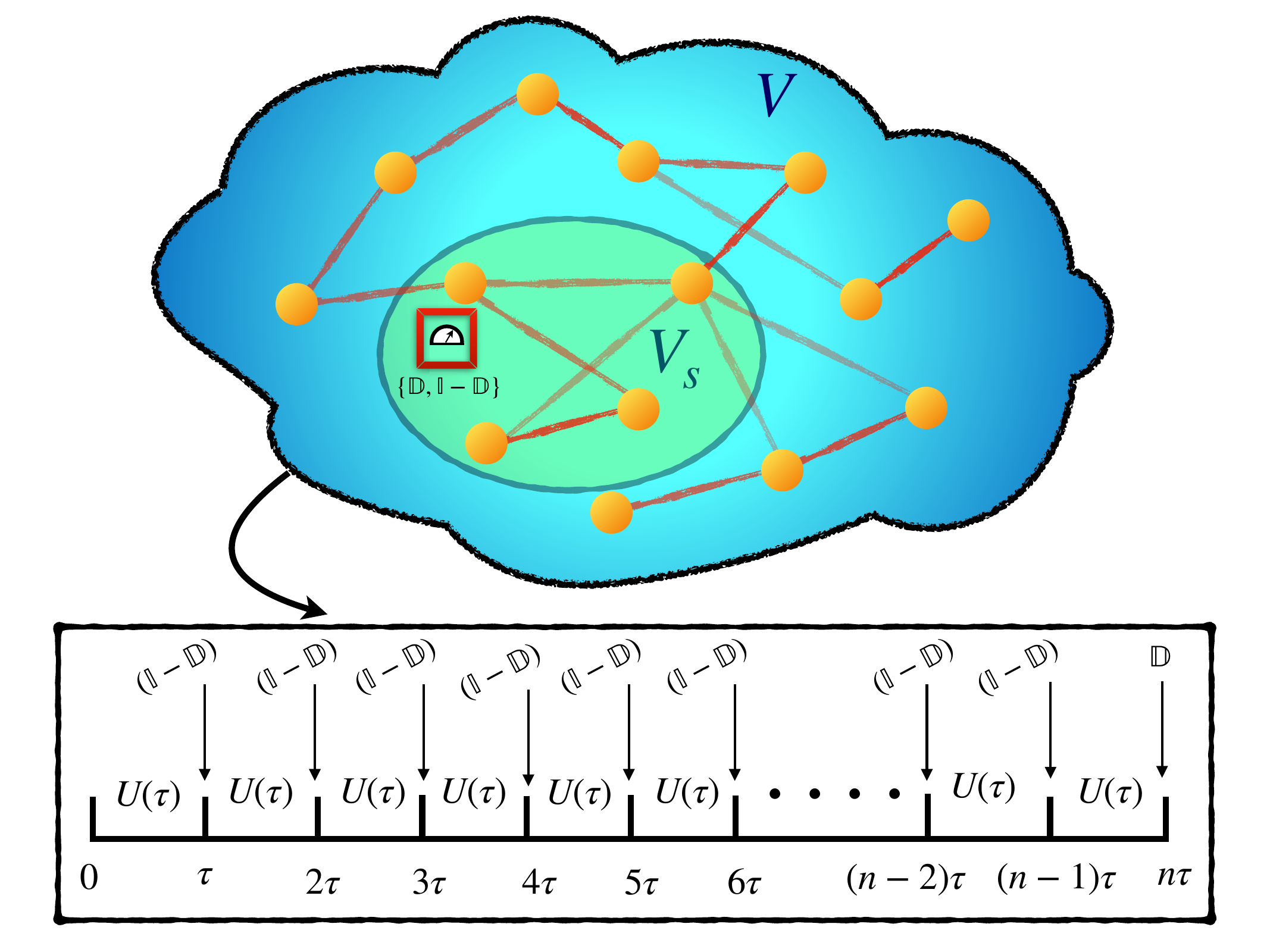}
        \caption{Schematic diagram of first detection of the walker in subspace $V_s$ of a finite graph, consisting of a set of vertices $V$ with measurement-induced quantum walk. In the MIQW protocol, the first hitting time statistics $F_n$ is captured through the application of unitary dynamics, $U(\tau)$, punctuated intermittently by measurements $\{\mathbb D, \mathbb{I-D}\}$ taken in a stroboscopic fashion with time interval $\tau$. Here, the $\mathbb D$ is the detector corresponding to the subspace $V_s$. The schematic shows that upto $(n-1)$th round of measurements, no-click event is occurring, i.e., the subspace $V_s$ is not detected, while  after the $n$th round of measurement, the particle is detected in subspace $V_s$ which corresponds to a successful click event. After successful detection the protocol is stopped.}
        \label{stroboscopic-diagram}
    \end{figure}

We employ the concepts of dark and bright energy levels to address the problem of detecting a particle in a subspace (see Fig. \ref{stroboscopic-diagram} for schematics), going beyond locating it in a single site. It has also been addressed by using different methods, namely non-Hermitian approach \cite{Dhar2015} and by establishing its connection with the properties of Schur function \cite{Grunbaum-schur-func2014}. It is crucial to highlight that in certain cases, the method proposed here can explain the total detection probability in a more simpler manner than the existing methods. 
In this work, we adopt the rank-nullity theorem \cite{golub2013matrix} to establish a connection between the existence of dark and bright states by selecting a subspace from the set of vertices of a discrete, and finite graph in which  the particle has to be identified.
In contrast to the scenario encountered in a localized single-site detection within a graph, we exhibit that for each degenerate energy level, there can exist more than one bright energy state, in the subspace detection within a system.  Subsequently, we derive explicit formulae for the orthonormal states within the dark energy subspace and its corresponding complementary bright energy subspace.  We demonstrate that the total detection probability decreases monotonically with an increase in the number of dark states in the system
which is related to the increase of rank and the position of the detector. Importantly, we provide a necessary and sufficient condition for detecting a particle in a subspace with certainty in an arbitrary finite graph having discrete, bounded and degenerate spectrum independent of vertex-localized initial states which can be important in quantum computation.  We observe that  increasing the rank of the subspace and strategically placing detectors can minimize the divergences in average hitting time observed in the case of a single-site detection.

The paper is organized in the following manner. In Sec. \ref{sec:subspace}, the problem of subspace detection and the quantities of interest are discussed. In the context of a particle to be detected in a subspace, the notion of dark as well as bright states   and the criteria for unit detection probability are  presented in Sec. \ref{sec:br_vs_dark_energy}.  Sec. \ref{sec:matrix_method} illustrates another method for computing the total detection probability  based on the computation of matrices numerically while both the methods are applied on interacting systems with nearest-neighbor and next nearest-neighbor hopping in Sec. \ref{sec:interacting}. In Sec. \ref{sec:avnomeasure}, we study the pattern of the average number of measurement in detecting a particle in subspace while the results are summarized in Sec. \ref{sec:conclu}. 



\section{Stroboscopic subspace  detection protocol}
\label{sec:subspace}

Let us consider a quantum mechanical particle moving on a finite graph having a set of vertices $V$$=\left\{\ket{l}\Big|\sum_{l=1}^L\ketbra{l}{l}=\mathbb I\text{ and }\braket{l}{m}=\delta_{l,m}\right\}_{l=1}^L$, described by a time-independent Hamiltonian $H=-\sum_{l,m=1}^L \gamma_{lm} \ketbra{l}{m}$ where $\gamma_{lm}$ are constants. Thus, unitary dynamics of the initial state, $\ket{\phi(0)}$, leads to the evolved state at time, $t$, as $\ket{\phi(t)}=U(t)\ket{\phi(0)}=e^{-iHt}\ket{\phi(0)}$. In the context of hitting problem \cite{Krovi2006hitting,Varbanov2008} in MIQW, we are interested to determine the position of particle in a given subspace of $V$. Towards achieving the same, we perform repeated projective measurements with periodicity $\tau$, corresponding to the  subspace $V_s=\{\ket{d_i}\}_{i=1}^{\tilde r<L}$,  written as 
\begin{eqnarray}
    \bigg\{\mathbb{D} = \sum_{i=1}^{\tilde r<L} \ketbra{d_i}, \mathbb{I}-\mathbb{D}\bigg\},
    \label{eq:detector_meas}
\end{eqnarray}
where $\ket{d_i}$ can be any vertex of $V$ with $\braket{d_i}{d_j}=\delta_{i,j}$ and $\tilde r$ represents the rank of the detector, $\mathbb D$. Under the assumption that the measurements are performed instantaneously, we consider a sequence of measurements until the particle is detected. Therefore, if the particle remains undetected up to $(n-1)$ number of measurement rounds, the unnormalized resulting state just before the successful detection at round $n$ can be written as 
\begin{eqnarray}
    \ket{\phi(n\tau)}\equiv\ket{\phi(n)}=U(\tau) [(\mathbb{I}-\mathbb{D})U(\tau)]^{n-1}\ket{\phi(0)}.
\end{eqnarray}
The first detection probability, i.e., the probability in detecting the particle for the first time after $n$th measurement attempt is given by \cite{Dhar2015, FriedmanKessler2017}
\begin{equation}
    F_{n} =  \bra{\phi(n)}\mathbb{D}\ket{\phi(n)},
    \label{Fn for subspace}
\end{equation}
while the \textit{total first detection probability}, $P_{\det}$ of the particle is defined as the detection probability after an infinite number of measurements conditioned on the fact that once the particle is detected, measurement process is stopped \cite{Barkai2020DarkState,FriedmanKessler2017}. Alternatively, we call it as \textit{total detection probability}, and mathematically, we can write it as 
\begin{eqnarray}
    P_{\det}=\sum_{n=1}^{\infty} F_n.
    \label{eq:pdet}
\end{eqnarray}
Also, the probability of the particle surviving the first $n$ rounds of measurement can be written as
\begin{eqnarray}
    \nonumber S_n &=& 1-\sum_{k=1}^n F_k\\\nonumber &=& \bra{\phi(n)} (\mathbb{I-D}) \ket{\phi(n)} \\
    \nonumber &=& \bra{\phi(0)}[U^{\dagger}(\tau)(\mathbb{I-D})]^n[(\mathbb{I-D})U(\tau)]^n \ket{\phi(0)} \\ &=& \bra{\phi(0)} \mathbb{S}^{{\dagger}^n} \mathbb{S}^n \ket{\phi(0)},
\end{eqnarray}
where $\mathbb{S}\equiv(\mathbb{I}-\mathbb{D})U(\tau)$ is the survival operator. Therefore, the final survival probability (i.e., in  $\lim_{n\to\infty}S_n$) reads as 
\begin{eqnarray}
    P_{\sur} = \lim_{n\to\infty} S_n = 1-P_{\det}.
\end{eqnarray}
In the case of identifying a particle in a fixed subspace, we will be focusing on developing a framework that can be utilized to obtain total detection probability.

\section{Prescription for calculation of total detection probability through dark and bright energy subspace }
\label{sec:br_vs_dark_energy}

We now develop a method that leads to a definite conclusion about whether a particle resides in a given region. In particular, we investigate the trends of the total detection probability, $P_{\det}$, by varying the subspace in which the particle is to be detected. To address this question, we provide a framework aimed at partitioning the energy space of the Hamiltonian into two distinct orthogonal subspaces, namely dark and bright subspaces \cite{Barkai2020DarkState}.

{\it Dark and bright states.} Given an initial state $\ket{\phi(0)}$, $P_{\det}(\phi(0))=0$, i.e., $F_n=0\forall n$ represents a \textit{dark state} with respect to a detection space $\mathbb{D}$ on the other hand, $P_{\det}(\phi(0))=1$ corresponds to \textit{bright state} which is detected with certainty. However, there can be initial states that are neither completely dark nor bright, and the first detection probability lies between $0$ and $1$, i.e., $0 < P_{\det} < 1$. 

We are interested in the stationary dark states, which are the eigenstates of both the unitary evolution $U(\tau)$ and survival operator $\mathbb{S}$. 
Let us denote the $k$-th energy level of the Hamiltonian, $H$ as $E_k$ and the corresponding set of eigenvectors as $\{\ket{E_{k,m}}\}_{m=1}^{g_k}$ where $g_k$ is the degeneracy of $E_k$. If an energy level is non-degenerate, we omit the index $m$.
We consider two different scenarios in case of degeneracy of $E_k$ while finding conditions for dark state to exist.  

{\it (i) Non-degenerate energy levels.} According to the definition, a non-degenerate energy level $E_k$ is a dark state if 
\begin{eqnarray}
    \mathbb{D}\ket{E_k} = 0,
    \label{eq:dark_subspace}
\end{eqnarray}
and
\begin{eqnarray}
    (\mathbb{I} - \mathbb{D})U(\tau)\ket{E_k} = \exp(-i E_k \tau)\ket{E_k}.
\end{eqnarray}
The condition in Eq.~(\ref{eq:dark_subspace}) for a non-degenerate energy eigenstate to be a dark state can be equivalently expressed as $\braket{d_i}{E_k}=0$ $\forall i$. 
In the other scenario, i.e., for degenerate energy levels, the physics of dark and bright states with respect to subspace detection is much more captivating as discussed below.

{\it (ii) Degenerate energy levels.} In case of degenerate eigenstates, $\{\ket{E_{k,m}}\}_{m=1}^{g_k}$, we construct an projector $\mathbb{E}_k$ which can be written mathematically as
\begin{eqnarray}
    \mathbb{E}_k = \sum_{m=1}^{g_k} \ketbra{E_{k,m}}.
    \label{Energy-projector}
\end{eqnarray}
Here, the first index $k$ in $\ket{E_{k,m}}$ corresponds to the distinct energy level, and the second index $m$ indicates the level of degeneracy present in each level. To find out the existence of dark states in the corresponding degenerate subspace, let us write any dark state in that particular subspace as 
\begin{eqnarray}
     \ket{\zeta_k} = \sum_{m=1}^{g_k} \alpha_{m}\ket{E_{k,m}}.
     \label{degen-dark-state}
\end{eqnarray}
Therefore, by Eq. \eqref{eq:dark_subspace} we get 
\begin{eqnarray}
    \mathbb{D}\ket{\zeta_k} &=& 0
    \label{eq:d_zetak}\\\nonumber
    \implies \sum_{i=1,m=1}^{\tilde r,g_{k}} \alpha_{m} \ket{d_{i}}\braket{d_{i}}{E_{k,m}} &=& 0,\\ \nonumber
    \implies \sum_{m=1}^{g_{k}} \alpha_{m} \ket{d_{i}}\braket{d_{i}}{E_{k,m}} &=& 0 \forall i\\
    \implies \mathcal{A}_k \ket{\tilde\alpha} &=& 0,
    \label{eq:condition_dark}
\end{eqnarray}
with the $\tilde{r}\times g_k$ matrix, $\mathcal{A}_k$, being
\begin{eqnarray}
        \mathcal{A}_k = \begin{bmatrix}
            \braket{d_1}{E_{k,1}} & \braket{d_1}{E_{k,2}} & \hdots &\braket{d_1}{E_{k,g_{k}}} \\
            \braket{d_2}{E_{k,1}} & \braket{d_2}{E_{k,2}} & \hdots &\braket{d_2}{E_{k,g_{k}}} \\
            \vdots & \vdots & \ddots & \vdots \\
            \braket{d_r}{E_{k,1}} & \braket{d_r}{E_{k,2}} & \hdots &\braket{d_r}{E_{k,g_{k}}} \\
        \end{bmatrix},
         \label{A-matrix}
\end{eqnarray}
and $\ket{\tilde\alpha} = (\alpha_1,\alpha_2,\ldots,\alpha_{g_k})^{T}$ being the coefficient vector of the dark state (see Eq.~\eqref{degen-dark-state}).
The above condition clearly shows that the existence of $\ket{\zeta_k}$ depends on the overlaps of $\ket{E_{k,m}}$ with $\ket{d_i}$, i.e, $\braket{d_{i}}{E_{k,m}}$ $\forall i,m$. Let us proceed to analyze this matter through a systematic examination of individual cases.

{\it Case I. }Consider the scenario when $\braket{d_{i}}{E_{k,m}}=0 \text{  }\forall \text{  } i,m $. In this case, all the degenerate energy eigenstates  $\{\ket{E_{k,m}}\}_{m=1}^{g_k}$ corresponding to energy $E_k$ are the dark states, i.e., $\ket{\zeta_k^j}=\ket{E_{k,j}}$ for $j=1,2,\ldots, g_k$. Therefore, the entire energy subspace is dark.

{\it Case II. }Let us consider a situation when $\braket{d_i}{E_{k,m}}=0$ for some of $i$ and $m$ but not all of them. We present one of our main findings as Theorem \ref{prop:dark_number}, from which the number of dark states in a system can be calculated. Note that in the complementary subspace to the dark subspace, the energy states are eventually bright states which will be proved later in Proposition \ref{prop:bright_pdet1}.
\begin{theorem}
    Number of dark states in the subspace $\{\ket{E_{k,m}}\}_{m=1}^{g_k}$ is equal to the dimension of the null space of $\mathcal{A}_k$, denoted as $\dim({\mathcal{N_A}_k})$ while the number of bright states is equal to the rank of matrix $\mathcal{A}_k$, $\rank(\mathcal{A}_k)$.
    \label{prop:dark_number}
\end{theorem}
\begin{proof}
It is evident from Eq.~(\ref{eq:condition_dark}) that the existence of dark states is equivalent to finding non-trivial solutions $(\text{trivial solution is }\alpha_i=0\text{ for }i=1,2,\ldots,g_k)$ of the matrix $\mathcal{A}_k$ which, in turn, is linked to the determination of its nullspace, denoted as $\mathcal{N}_{\mathcal{A}_k}$. Therefore, the number of dark states in the corresponding energy subspace is equal to dimension of the nullspace, i.e., $\dim{(\mathcal{N}_{\mathcal{A}_k})}$. Moreover, from rank-nullity theorem \cite{golub2013matrix}, we know 
\begin{eqnarray}
    \nonumber \rank(\mathcal{A}_k) + \dim(\mathcal{N_A}_k) = \text{Number of columns of $\mathcal{A}_k$} = g_k.\\
    \label{Rank-null}
\end{eqnarray}
Therefore, the number of bright states is just $\rank(\mathcal{A}_k)$ since any degenerate subspace is spanned by dark states and its complement space, containing only bright states \cite{Barkai2020DarkState}. Mathematically, we can write that $\{\ket{E_{k,m}}\}_{m=1}^{g_k}$ is spanned by $\left\{\left\{\ket{\zeta_k^j}\right\}_{j=1}^{\dim({\mathcal{N_A}_k})},\left\{\ket{\eta_k^j}\right\}_{j=1}^{\rank({\mathcal{A}_k})}\right\}$, i.e., following Eq.~\eqref{Energy-projector}, it can be expressed as
\begin{eqnarray}
    \mathbb{E}_k = \sum_{j=1}^{\dim(\mathcal{N_A}_k)}\ketbra{\zeta_k^j}{\zeta_k^j} + \sum_{j=1}^{\rank(\mathcal{N_A}_k)}\ketbra{\eta_k^j}{\eta_k^j}.
    \label{eq:e_k_proj}
\end{eqnarray}
\end{proof}

Let us now explicitly calculate the basis states consisting of bright and dark states for degenerate energy levels. The general form of the matrix $\mathcal{A}_k$ after performing row reduction on it and removing zero rows can be updated as 
\begin{eqnarray}
    \mathcal{A}_k = \begin{bmatrix}
            a_{1,1} & a_{1,2} & \hdots & a_{1,l_k} & \hdots &a_{1,g_k} \\
            0 & a_{2,2} & \hdots & a_{2,l_k} &\hdots &a_{2,g_k} \\
            \vdots & \vdots & \ddots & \vdots & \vdots \\
            0 & 0 & \hdots& a_{l_k,l_k} & \hdots &a_{l_k,g_k}
        \end{bmatrix}_{l_k\times g_k},
         \label{eq:row_reduced_A-matrix}
\end{eqnarray}
of reduced dimension where $l_k=\rank(\mathcal{A}_k)\leq g_k$ with $a_{i,j}=0$ $\forall i>j$. Also from rank-nullity theorem, we know that the dark subspace corresponding to a degenerate energy $E_k$ exists if $l_k<g_k$. Now we can write one of the dark states as
\begin{eqnarray}
    \ket{\zeta_{k}^1} = N_1\begin{vmatrix}
        \ket{E_{k,1}} & \ket{E_{k,2}} & \hdots & \ket{E_{k,l_k}}&\ket{E_{k,l_k+1}} \\
        a_{1,1} & a_{1,2} & \hdots & a_{1,l_k} & a_{1,l_k+1} \\
            0 & a_{2,2} & \hdots & a_{2,l_k}  &a_{2,l_k+1} \\
            \vdots & \vdots & \ddots & \vdots & \vdots \\
            0 & 0 & \hdots& a_{l_k,l_k}  &a_{l_k,l_k+1}
  \end{vmatrix},
\end{eqnarray}
while all other dark states can be iteratively written as
\begin{widetext}
\begin{eqnarray}
    \nonumber \ket{\zeta_{k}^j} = 
    N_j\begin{vmatrix}
        \ket{E_{k,1}} & \ket{E_{k,2}} & \hdots & \ket{E_{k,l_k}}&\ket{E_{k,l_k+1}} & \ket{E_{k,l_k+2}}  & ... & \ket{E_{k,l_k+j}} \\
        a_{1,1} & a_{1,2} & \hdots & a_{1,l_k} & a_{1,l_k+1} &a_{1,l_k+2} & ... & a_{1,l_k+j} \\ 
        0 & a_{2,1}  &  ... & a_{2,l_k} & a_{2,l_k+1} &a_{2,l_k+2} & ... & a_{2,l_k+j} \\
        \vdots & \vdots & \ddots &\vdots &\vdots & \vdots & \ddots &\vdots\\
        0 & 0 & \hdots& a_{l_k,l_k}  &a_{l_k,l_k+1} & a_{l_k,l_k+2} & ... & a_{l_k,l_k+j} \\
        \braket{\zeta^{1}_k}{E_{k,1}} & \braket{\zeta^{1}_k}{E_{k,2}} & ... & \braket{\zeta^{1}_k}{E_{k,l_k}} &\braket{\zeta^{1}_k}{E_{k,l_k+1}} & \braket{\zeta^{1}_k}{E_{k,l_k+2}} & ... & \braket{\zeta^{1}_k}{E_{k,l_k+j}} \\
        \vdots & \vdots &\ddots &\vdots & \vdots &\vdots & \ddots & \vdots\\ 
        \braket{\zeta^{j-1}_k}{E_{k,1}} & \braket{\zeta^{j-1}_k}{E_{k,2}} & \hdots & \braket{\zeta^{j-1}_k}{E_{k,l_k}} & \braket{\zeta^{j-1}_k}{E_{k,l_k+1}} & \braket{\zeta^{j-1}_k}{E_{k,l_k+2}} & \hdots & \braket{\zeta^{j-1}_k}{E_{k,l_k+j}}
  \end{vmatrix}_{(l_k+j)\times (l_k+j)}.\\
\end{eqnarray}
\end{widetext}
with $N_j$ being the normalization constant. 

On the other hand, the projector of $E_k$ sector acting on the individual states of the detector $\left\{\ket{d_i}\right\}_{i=1}^{\tilde{r}}$ can be written as $\left\{ \frac{\mathbb{E}_k\ket{d_i}}{\sqrt{\bra{d_i}\mathbb{E}_k \ket{d_i}}} \right\}_{i=1}^{\tilde{r}}$, which, by Gram-Schmidt orthogonalization procedure, can be transformed to mutually orthogonal set as
\begin{eqnarray}
    \left\{\ket{\eta_k^j}=\sum_{i=1}^{\tilde r}c_{k,i}^j\mathbb{E}_k\ket{d_i}\Bigg|\braket{\eta_k^j}{\eta_k^{j'}}=\delta_{j,j'}\right\}_{j=1}^{l_k}.
    \label{eq:bright_mus}
\end{eqnarray}
Note that this is equivalent in orthonormalizing the set $\left\{\sum_{m\geq p}^{g_k}a_{p,m}\ket{E_{k,m}}\right\}_{p=1}^{l_k}$ as evident from Eqs. \eqref{A-matrix} and \eqref{eq:row_reduced_A-matrix}. Eventually, the set $\left\{\ket{\eta_k^j}\right\}_{j=1}^{l_k}$ is actually the bright states corresponding to $E_k$ sector which will be proved shortly.

\subsection{Detection probability from bright or dark space projection}
We possess the requisite foundation to calculate the total detection probability by exploiting the idea of dark and bright states discussed above. For an initial state$\ket{\phi(0)}$, we can write 
\begin{eqnarray}
    \ket{\phi(0)} = \mathbb{P}_{\mathcal{H}_{\zeta}}\ket{\phi(0)} + \mathbb{P}_{\mathcal{H}_{\eta}}\ket{\phi(0)},
    \label{eq:phiin}
\end{eqnarray}
where $\mathbb{P}_{\mathcal{H}_{\zeta}} = \mathlarger\sum_{k,j}\ketbra{\zeta_k^j}{\zeta_k^j}$ and $\mathbb{P}_{\mathcal{H}_{\eta}} = \mathlarger\sum_{k,j}\ketbra{\eta_k^j}{\eta_k^j}$ are the projectors of dark and bright subspaces respectively with $\mathbb{P}_{\mathcal{H}_{\zeta}}+\mathbb{P}_{\mathcal{H}_{\eta}}=\mathbb{I}$. The survival probability \cite{krovi2006hypercube} can be found by considering the overlap of the initial state with the dark subspace, i.e., $P_{\sur} =\bra{\phi(0)}\mathbb{P}_{\mathcal{H}_{\zeta}}\ket{\phi(0)}$. Consequently, from Eq. \eqref{eq:phiin}, it follows that the total detection probability can be written as \cite{Barkai2020DarkState}
\begin{eqnarray}
    P_{\det} &=& \bra{\phi(0)}\mathbb{P}_{\mathcal{H}_{\eta}}\ket{\phi(0)} \label{eq:pdet_bright1}\\ &=& \mathlarger{\sum}_{k}\mathlarger{\sum}_{j=1}^{l_k}\mathlarger{|}\bra{\eta_k^j}\ket{\phi(0)}\mathlarger{|}^2,
    \label{eq:pdet_bright}
\end{eqnarray}
where the index $k$ runs over all the distinct energy levels of $H$, responsible for the evolution of the system, and $l_k=\rank{(\mathcal{A}_k)}$. Equivalently, one can also calculate 
\begin{eqnarray}
    P_{\det} = 1 - \mathlarger{\sum}_{k}\mathlarger{\sum}_{j=1}^{\dim(\mathcal{N_A}_k)}\mathlarger{|}\bra{\zeta_k^j}\ket{\phi(0)}\mathlarger{|}^2.
    \label{eq:pdet_dark}
\end{eqnarray}
Having formulated the detection probability $P_{\det}$, let us now show that any $\ket{\eta_k^j}$ represents the bright state as mentioned earlier. 
\begin{proposition}
    Any state from the set $\left\{\ket{\eta_k^j}\right\}_{j=1}^{\rank(\mathcal{A}_k)}$, resides in the complementary space of the dark subspace, having unit total first detection probability, i.e., $P_{\det}=1$.
    \label{prop:bright_pdet1}
\end{proposition}
\begin{proof}
    For any $\ket{\eta_k^j}$, we obtain:
    \begin{eqnarray}
    \mathbb{P}_{\mathcal{H}_{\zeta}}\ket{\eta_{k}^j} &=& \sum_{k'}\sum_{j'=1}^{g_{k'} - l_{k'}} \ket{\zeta_{k'}^{j'}}\bra{\zeta_{k'}^{j'}}\ket{\eta_{k}^j} \nonumber \\
        &=&\sum_{k'}\sum_{j'=1}^{g_{k'} - l_{k'}} \ket{\zeta_{k'}^{j'}} \delta_{k'k}\bra{\zeta_{k}^{j'}}\ket{\eta_{k}^j} \nonumber \\
        &=& \sum_{j'=1}^{g_{k} - l_{k}} \ket{\zeta_{k}^{j'}}\bra{\zeta_{k}^{j'}}\ket{\eta_{k}^j} \nonumber \\
        &=& \sum_{j'=1}^{g_{k} - l_{k}} \ket{\zeta_{k}^{j'}}\bra{\zeta_{k}^{j'}} \sum_{i=1}^{\tilde r} c_{k,i}^j \mathbb{E}_k \ket{d_i} \nonumber \\
        &=& \sum_{j'=1}^{g_{k} - l_{k}} \sum_{i=1}^{\tilde r} c_{k,i}^j \ket{\zeta_{k}^{j'}}\braket{\zeta_{k}^{j'}}{d_i} \nonumber \\
        &=& 0.
    \end{eqnarray}
Here, the second to third line is followed from Eq.~\eqref{eq:bright_mus}, and the last line is due to the definition of dark states (see Eq.~\eqref{eq:d_zetak}). Therefore, the survival probability, $P_{\sur}=0$ and consequently $P_{\det}=1$ corresponding to $\ket{\eta_k^j}\forall{j,k}$.
\end{proof}
Before calculating $P_{\det}$ for specific system configurations, we shall discuss some generic features.
\begin{proposition}
Independent of the existence of dark states in the system, if the initial state can be written as a linear combination of only bright energy states of the systems, i.e., $\ket{\phi(0)}=\sum_{j,k}c_{k}^j \ket{\eta_k^j}$ for any value of $c_k^j$ $\forall j,k$, such that $\sum_{j,k}|c_k^j|^2=1$, the total first detection probability, $P_{\det}$, is unity for that initial state.
\label{prop:phin_bright}
\end{proposition}
\begin{proof}
    Since the initial state is a linear combination of the bright states only, Eq.~\eqref{eq:phiin} reduces to 
    $\ket{\phi(0)}=\mathbb{P}_{H_{\eta}}\ket{\phi(0)}$ where $c_k^j=\braket{\eta_k^j}{\phi(0)}$ and following Eq.~\eqref{eq:pdet_bright1}, we can write $P_{\det}=\braket{\phi(0)}{\phi(0)}=1$. 
\end{proof}
Furthermore, when the initial state is the linear combination of the vectors in the detector subspace, then measurements detect the return of the particle in the subspace defined as the return problem \cite{Grunbaum-schur-func2014}. Therefore, if the initial state is $\ket{\phi(0)}=\sum_{i=1}^{\tilde r} e_i \ket{d_i}$,  $\mathbb{P}_{\mathcal{H}_{\zeta}}\ket{\phi(0)}=0$ from the definition of dark states  which means $\ket{\phi(0)}=\mathbb{P}_{H_{\eta}}\ket{\phi(0)}$ and from Proposition~\ref{prop:phin_bright}, we immediately obtain the following corollary:
\begin{corollary}
    For initial states which are linear combination of detector states (as mentioned in Eq.~\eqref{eq:detector_meas}), i.e., $\ket{\phi(0)}=\sum_{i=1}^{\tilde r} e_i \ket{d_i}$, termed as a return problem, $P_{\det}=1$.
    \label{cor:return}
\end{corollary}
Apart from the initial state, the dependency of $P_{\det}$ on the rank and position of the detector subspace can be assessed from the study of dark and bright states. From Eq.~\eqref{eq:pdet_dark}, it is evident that if no dark state exists in a system, $P_{\det}$ is surely unity, independent of any initial state. Moreover, the existence of the dark states is related to the features of $\mathcal{A}_k$ as mentioned in Proposition~\ref{prop:dark_number}. We will now establish a connection between the characteristics of the detection subspace and the deterministic nature of the measurement-induced quantum walk.
\begin{theorem}
    Let $H$ be a discrete, bounded, degenerate Hamiltonian of the finite graph defined by vertices, $V=\left\{\ket{l}|\sum_{l=1}^L\ketbra{l}{l}=\mathbb I\text{ and }\braket{l}{m}=\delta_{l,m}\right\}_{l=1}^L$ which drives a system periodically in measurement induced quantum walk. Suppose $H$ has degeneracy $g_k$ corresponding to energy level $E_k$. By considering a subspace $V_s = \left\{\ket{d_j}\right\}_{j=1}^{\tilde r<L}$ of $V$ where $\ket{d_j}$ can be any $\ket l$ with $\braket{d_i}{d_j}=\delta_{i,j}$, its detection  $\mathbb D = \sum_{j=1}^{\tilde r}\ketbra{d_j}{d_j}$ is performed deterministically, i.e., $P_{\det}=1$ independent of any localized initial state $\ket{\phi(0)}=\ket{l}\in V$ if and only if both the conditions are satisfied: 
    \begin{enumerate}
        \item[C1.]  for each non-degenerate $E_k$, $\braket{d_j}{E_{k}}\neq 0$, at least for one $d_j\in V_s$, and
        \item[C2.] for each degenerate $E_k$, $\rank(\mathcal{A}_k)=g_k$ which is possible when $\tilde r\geq \max_{k}g_k$.
    \end{enumerate}
    are satisfied.
    \label{th1}
\end{theorem}
\begin{proof}
    If \textit{C1} is satisfied, from definition (Eq.~\eqref{eq:dark_subspace}), the non-degenerate energy subspace has no dark states. From \textit{C2}, if rank of the detector, $\tilde r\geq \max_k g_k$ and by performing row reduction on $\mathcal{A}_k$ we find $\rank{\mathcal{A}_k}=g_k$ for each degenerate energy $E_k$, $\dim{(\mathcal{N}_{\mathcal{A}_k})}=0\forall k$. In this case, it follows from Proposition~\ref{prop:dark_number} that there exists no dark states corresponding to degenerate energy subspace of the system. Therefore, $P_{\det}=1$ as evident from Eq.~\eqref{eq:pdet_dark}.

    Now, we concentrate on the case when $P_{\det}=1$ independent of the initial localized  states of the corresponding system it implies \textit{C1} and \textit{C2}. Mathematically, this can be written as $\bra{l}\mathbb{P}_{\mathcal{H}_{\eta}}\ket{l}=1$ $\forall l\in V$ which follows from Eq.~\eqref{eq:pdet_bright1}. Now, as $\left\{\ket{l}\right\}_{l=1}^L$ forms an orthonormal basis, it must be the case that $\mathbb{P}_{\mathcal{H}_{\eta}}=\mathbb I$ which means the system has no dark states. Therefore, from Proposition~\ref{prop:dark_number}, we can see that $\rank(\mathcal{A}_k)=g_k$ for each degenerate energy subspace, and hence \textit{C2} is true. Moreover, the nonexistence of dark states in non-degenerate energy subspace implies \textit{C1} is obvious from Eq.~\eqref{eq:dark_subspace}.
\end{proof}

\section{Alternative method for obtaining total detection probability}
\label{sec:matrix_method}

Without delving into the detailed physics of dark and bright states, we can obtain $P_{\det}$ solely through the utilization of Eq.~\eqref{eq:pdet} by saturating the summation to a finite value. However, this method is computationally inefficient and can be a time-consuming affair. To overcome this issue, we propose a reformulation of $P_{\det}$ for the detection of a subspace $V_s$ by detector $\mathbb D$, following the methodologies outlined in Ref. \cite{Varbanov2008} and \cite{Kessler2021}. 

Before laying out the result, let us define few matrices $S,W,T_j$ of dimension $L\times L$. Firstly, the energy spectrum of $H$ is given by $\{\epsilon_p,\ket{\epsilon_p}\}_{p=0}^{L-1}$ where $\epsilon_p$ is the eigenvalue of $H$ with an eigenvector $\ket{\epsilon_p}$. Now, the elements of the aforementioned matrices are given by $S_{kp} = \delta_{kp}\braket{\epsilon_p}{\phi(0)}$, $W_{kp}=1$, $(T_{j})_{kp}=\delta_{kp}\braket{\epsilon_p}{d_j}$ and $C=\mathbb{I} - \sum_{j=1}^{\tilde r}T_{j}^{*}WT_{j}$. The unitary evolution operator $U$ can be written in the energy eigenbasis as $U_{kp} = \delta_{kp}e^{-i \epsilon_p \tau}$. Moreover, for any operator $O$, we characterize $\bar{O}\equiv O^*\otimes O$. Finally, we define $\mathcal{L} = \mathbb{I} - \bar{U}\bar{C}$ which is an $L^2\times L^2$ dimensional matrix.

\begin{proposition}
    In a finite graph of $L$ vertices, the probability of first successful detection after infinite number of measurements (with periodicity $\tau$) in a subspace of dimension $\tilde r<L$ can be written as $P_{\det} = \sum_{i=1}^{\tilde r} \Tr\left[ \mathcal{L}^{-1} \bar{U} \bar{S} \bar{W} \bar{T_{i}^{*}} \right]_{L^2\times L^2}$.
\end{proposition}
\begin{proof}
    First, we show that Eq.~\eqref{Fn for subspace} can be rewritten as a trace of product of $L^2\times L^2$ matrices of the form,
\begin{eqnarray}
    F_n = \sum_{i=1}^{\tilde r}\Tr\left[ (\bar{U} \bar{C})^{n-1} \bar{U} \bar{S} \bar{W} \bar{T_{i}^{*}} \right]_{L^2\times L^2},
    \label{eq:fn}
\end{eqnarray}
(see Appendix \ref{Appndx:Subspace}). Finally, using Eq.~\eqref{eq:pdet}, the probability of first successful detection is found to be 
\begin{eqnarray}
    \nonumber P_{\det}
        &=& \sum_{n=1}^{\infty} \sum_{i=1}^{\tilde r} \Tr\left[ (\bar{U}\bar{C})^{n-1} \bar{U} \bar{S} \bar{W} \bar{T_{i}^{*}} \right] \\ 
        &=& \sum_{i=1}^{\tilde r} \Tr\left[ \mathcal{L}^{-1} \bar{U} \bar{S} \bar{W} \bar{T_{i}^{*}} \right].
        \label{eq:mat_pdet}
\end{eqnarray}
Note that instead of computing the inverse of $\mathcal{L}$, we need to calculate pseudoinverse \cite{golub2013matrix}, $\mathcal{L}^{-1}$ due to the fact that $\mathcal{L}$ becomes singular when the spectrum of $U(\tau)$ is degenerate. However,  by definition pseudoinverse reduces to traditional inverse when $\mathcal{L}$ is non-singular. Although we perform the measurements periodically (with fixed $\tau$), it is clear that during measurements at random times, the above formulation can be very efficient \cite{Kessler2021}. 
\end{proof}

\section{Interacting system with moderate range hopping}
\label{sec:interacting}

Let us utilize the concepts, namely the dark and bright states, developed in Sec.~\ref{sec:br_vs_dark_energy} for a  system with nearest and next nearest-neighbor hopping. The initial states are localized on the graph nodes, and the rank of the detector is varied for detection of particle in subspaces of higher dimensions. The entire investigations will also highlight  the advantages and limitations of the methods discussed  in Secs. \ref{sec:br_vs_dark_energy} and \ref{sec:matrix_method}. 

\subsection{Total detection probability of subspace for nearest-neighbor interacting system}

Let us consider a nearest-neighbor (NN) interacting lattice of $L$ sites with periodic boundary condition described by the Hamiltonian,
\begin{eqnarray}
    H = -\gamma \sum_{i=0}^{L-1} \ketbra{i+1}{i} + \ketbra{i}{i+1},
    \label{H for NN}
\end{eqnarray}
where $L$ is an even integer and set $\gamma=1$ without loss of generality. The eigenvalues and eigenvectors of the above system are given by \cite{Dhar2015}
\begin{eqnarray}
    \nonumber \epsilon_p &=& -2\gamma  \cos(\frac{2\pi p}{L}), \,\text{and}  \\
    \ket{\epsilon_p} &=& \frac{1}{\sqrt{L}} \sum_{j=0}^{L-1} \exp(\frac{i2\pi p j}{L})\ket j,
    \label{NN evals and evecs}
\end{eqnarray}
respectively where $p=0,1,\ldots,(L-1)$. Notice that only $\ket{\epsilon_0}$ and $\ket{\epsilon_{L/2}}$ are non-degenerate eigenstates. By relabeling the spectrum, we can write the non-degenerate subspace as $\left\{E_k,\hspace{0.1cm} \ket{E_{k}}|k=0,L/2\right\}$ whereas the doubly degenerate subspace is $\left\{E_k,\hspace{0.1cm} \left\{\ket{E_{k,m}}\right\}_{m=1}^2|k=1,2,\ldots,\frac{L}{2}-1\right\}$ with
\begin{eqnarray}
    \nonumber \ket{E_{0}}=\ket{\epsilon_0}; \hspace{0.3cm} \ket{E_{L/2}} = \ket{\epsilon_{L/2}},\\ \nonumber
    \left. \begin{aligned}
        \ket{E_{k,1}}&=\ket{\epsilon_k}\\
        \ket{E_{k,2}}&=\ket{\epsilon_{L-k}}
    \end{aligned}
    \right\} \text{  for }k = 1,2,\ldots,\frac{L}{2}-1.\\
\end{eqnarray}

We take the initial state as a localized state in a position basis, $\ket s$.  Let us now vary the rank of the detector $\mathbb{D}$, denoted by $\tilde{r}$ which belongs to the position basis. Note that $\braket{E_0}{d}\neq0$ and $\braket{E_{L/2}}{d}\neq0$ where $\ket{d}$ is any localized position state with the position being $d=0,\ldots,L-1$  of the lattice. Such observation suggests that irrespective of $\tilde r$, non-degenerate energy subspace is completely bright, i.e., $\ket{\eta_0} = \ket{\epsilon_0}$ and  $\ket{\eta_{L/2}} = \ket{\epsilon_{L/2}}$ are bright states.

\subsubsection{Rank-1 detection of particle ($\tilde{r}=1$)}
Let us study the problem where a rank-$1$ detector state,  $\mathbb{D}=\ketbra{d}$, is used to detect the particle which is in some localized position in the lattice  \cite{Barkai2020DarkState}. In the rank-$1$ scenario, there can be only a single bright state corresponding to each distinct energy $E_k$ for $k=0,1,\ldots, L/2$. Therefore, Eq.~\eqref{eq:pdet_bright1} corresponding to the first detection probability reduces to \cite{Barkai2020DarkState}
\begin{eqnarray}
    P_{\det} &=& \sum_{k} \frac{|\bra{s}\mathbb{E}_k\ket{d}|^2}{\bra{d}\mathbb{E}_k\ket{d}}.
    \label{eq:sed}
\end{eqnarray}
In case of doubly degenerate energy levels where $\epsilon_k = \epsilon_{L-k}$ with $k \ne 0, L/2$, the bright and the dark states respectively take the form as
\begin{eqnarray}
    \ket{\eta_{k}} &=& \frac{1}{\sqrt{2}} \left(e^{-\frac{i 2 \pi k d}{L}}\ket{\epsilon_k} + e^{\frac{i 2 \pi k d}{L}}\ket{\epsilon_{L-k}}\right),\, \text{and}
    \label{eq:bright_r1}\\
    \ket{\zeta_{k}} &=& \frac{1}{\sqrt{2}} \left(e^{-\frac{i 2 \pi k d}{L}}\ket{\epsilon_k} - e^{\frac{i 2 \pi k d}{L}}\ket{\epsilon_{L-k}}\right).
    \label{eq:dark_r1}
\end{eqnarray}
Finally, from both Eqs.~\eqref{eq:sed} and \eqref{eq:pdet_bright} we can find
\begin{eqnarray}
    P_{\det}( s) &=& \frac{2}{L} + \frac{2}{L}\sum_{k=1}^{\frac{L}{2}-1}\cos^2\left[\frac{2\pi k(d-s)}{L} \right]  \nonumber \\
                &=& \begin{cases}
                        1 &s = d, d+\frac{L}{2}, \\
                        1/2 &\text{otherwise}.
                    \end{cases}
    \label{eq:pdet_r1}
\end{eqnarray}
Note that $s=d$ corresponds to the return problem as discussed in Corollary~\ref{cor:return}. The case where the initial state is diametrically opposite to the detector, i.e., $s=d+\frac{L}{2}$, we can write
\begin{eqnarray}
   \nonumber \ket{s}&=&\ket{d+\frac{L}{2}} \\ \nonumber &=& \frac{1}{\sqrt{L}} \left[ \ket{\eta_0} + (-1)^{d} \ket{\eta_{L/2}} + \sqrt{2}\sum_{k=1}^{L/2 - 1} (-1)^{k} \ket{\eta_{k}}\right],\\
    \label{eq:special_eqn}
\end{eqnarray}
as a linear combination of bright energy states only which gives $P_{\det}=1$ following Proposition~\ref{prop:phin_bright}. Moreover, by fixing $L=10$, we calculate $P_{\det}$ numerically by Eq.~\eqref{eq:mat_pdet} without delving into the details of dark and bright states which matches with the above result.

\subsubsection{Identification of particle with rank-2 detector ($\tilde{r}=2$)}
Let us now increase the rank of $\mathbb{D}$ to be $2$, which can be written as
\begin{eqnarray}
    \mathbb D = \ketbra{d_1} + \ketbra{d_2}.
    \label{rank-2_arb}    
\end{eqnarray}
Due to the translation symmetry in $H$ (Eq.~\eqref{H for NN}), $P_{\det}$ must remain invariant by a constant shift, i.e., $P_{\det}(s,d_1,d_2)=P_{\det}(s+c,d_1+c,d_2+c)$. Therefore, we fix $d_1=0$ and vary $d_2$ from $1$ to $\frac{L}{2}$ .i.e., mathematically we can write
\begin{eqnarray}
    \mathbb D = \ketbra{0} + \ketbra{d_2}.
    \label{rank-2-D}
\end{eqnarray}
In case of doubly degenerate energy subspaces, we calculate the matrix $\mathcal{A}_k$ as
\begin{eqnarray}
\mathcal{A}_k &=& \begin{bmatrix}
    \braket{0}{\epsilon_k} & \braket{0}{\epsilon_{L-k}} \\ 
    \braket{d_2}{\epsilon_k} & \braket{d_2}{\epsilon_{L-k}} \\ 
\end{bmatrix} \\
&=& \frac{1}{L}\begin{bmatrix}
    1 & 1 \\
    \exp(\frac{i2\pi k d_2}{L}) & \exp(\frac{-i2\pi k d_2}{L}) \\
\end{bmatrix}.
\end{eqnarray}
Moreover, the row echelon form of $\mathcal{A}_k$ can be written as 
\begin{eqnarray}
    REF(\mathcal{A}_k)=\begin{bmatrix}
        1 & 1 \\
        0 & 2 i\sin\left(\frac{2\pi k d_2}{L}\right) \\  
    \end{bmatrix}.
    \label{eq:ref_r2_nn}
\end{eqnarray}
Let us now discuss the different cases of null space of $\mathcal{A}_k$.

 {\it Case I.} If $d_2 \ne \frac{m_k L}{2k}$ for some $ m_k\in \mathbb{Z}^+$ where $\mathbb{Z}^+$ denotes the set of positive integers, it is evident from Eq.~\eqref{eq:ref_r2_nn} that  $\rank(\mathcal{A}_k$) $=2$. Therefore, the nullspace is trivial corresponding to energy $E_k$. Consequently, following Theorem \ref{prop:dark_number}, no dark states exist in the  entire energy spectrum if the above condition satisfies for all $k$ which means $\mathbb{P}_{\mathcal{H}_\eta}=\sum_k \mathbb{E}_k =\mathbb I$. Therefore, the first detection probability is 
\begin{eqnarray}
    P_{\det} =  1 \text{ if }d_2\neq \frac{m_k L}{2k} \text{  } \forall k,
\end{eqnarray}
which clearly demonstrates the utility of Theorem \ref{prop:dark_number}.

{\it Case II.} In the case $d_2 = \frac{m_k L}{2k}$ for some $ m_k\in \mathbb{Z}^+$, we have
\begin{eqnarray}
\mathcal{A}_k &=& \frac{(-1)^{m_k}}{L}
\begin{bmatrix}
    1 & 1 \\
    1 & 1 \\
\end{bmatrix},
\end{eqnarray}
which have $\rank(\mathcal A_k)=1$ and $\dim(\mathcal {N_A}_k)=1$. Consequently, the dark and bright states are given by
\begin{eqnarray}
    \ket{\zeta_k} = \frac{1}{\sqrt{2}} \left(\ket{\epsilon_k} - \ket{\epsilon_{L-k}} \right),
    \label{eq:dark_r2}
\end{eqnarray}
and
\begin{eqnarray}
    \ket{\eta_k} = \frac{1}{\sqrt{2}} \left(\ket{\epsilon_k} + \ket{\epsilon_{L-k}} \right)
    \label{eq:bright_r2}
\end{eqnarray}
respectively. Specifically, for $d=L/2$ there exists $m_k$ such that $m_k=k \text{  }\forall k$. The bright state projector can be written as 
\begin{equation}
    \mathbb{P}_{\mathcal{H}_{\eta}} = \ketbra{E_0} + \ketbra{E_{L/2}} + \sum_{k=1}^{L/2 - 1} \ketbra{\eta_k}, 
\end{equation}
which gives 
\begin{eqnarray}
    \nonumber P_{\det}(s) &=& \frac{2}{L} + \frac{2}{L}\sum_{k=1}^{\frac{L}{2} - 1} \cos^{2}\left[ \frac{2\pi k s}{L} \right] \\ &=& \begin{cases}
                    1  & s= 0, L/2, \\
                    1/2 & \text{otherwise}.
              \end{cases}
    \label{eq:p_det_r2}          
\end{eqnarray}
Notice that in this case, $P_{\det}$, as shown in Eq.~\eqref{eq:p_det_r2} matches exactly with the first detection probability (Eq.~\eqref{eq:pdet_r1}) of rank- $1$ detector. The above result is a consequence of the fact that the dark and bright states (Eqs.~\eqref{eq:dark_r2} and \eqref{eq:bright_r2}) for $\mathbb D = \ketbra{0}{0}+\ketbra{L/2}{L/2}$ is same (upto an overall phase) as that for $\mathbb D = \ketbra{0}{0}$ or $\mathbb D = \ketbra{L/2}{L/2}$ which is given by Eqs.~\eqref{eq:bright_r1} and \eqref{eq:dark_r1}. To visualize the entire investigation, we carry out the analysis for a finite system-size, i.e., for fixed lattice sites, specifically $L=10$ and $L=20$.

\begin{figure}[h]
    \centering
    \includegraphics[width=\linewidth]{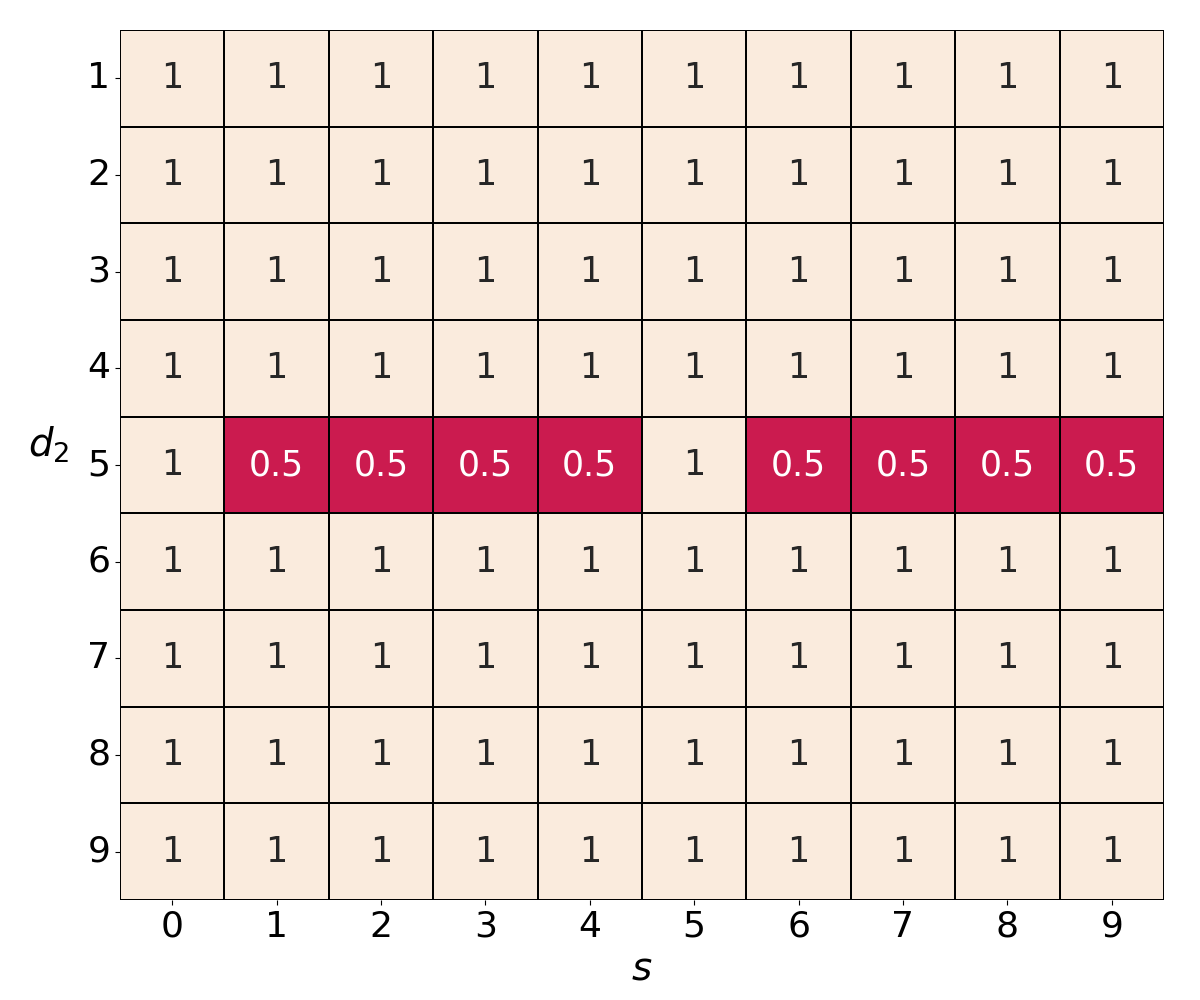}
    \caption{(Color online.) Total detection probability for detection of a particle in two-dimensional subspace for the cyclic graph $(L=10)$ as defined in Eq. (\ref{H for NN}) and detector as Eq. (\ref{rank-2-D}). Each column corresponds to fixed initial state $\ket{s}$ and detector $\mathbb{D} = \ketbra{0} + \ketbra{d_2}$ with $d_2 = 1,2,\cdots,9$. Each row corresponds to a fixed $\mathbb{D}$ ($d_2$ fixed) with varying initial state $\ket{s}$ where $s=0,1,\cdots 9$. The row with $d_2=5$ corresponds to detector $\mathbb{D}= \ketbra{0} + \ketbra{5}$ . The total detection probability for this subspace, $P_{\det}=0.5 \le 1$ for all initial states except the return problem (i.e., $s=0$ or $s=5$) due to the existence of dark states given by Eq.~\eqref{eq:dark_r2}. Except $d_2=5$, all other rows have $P_{\det}=1$ (as also seen from Theorem~\ref{th1}).}
    \label{fig:NN-rank-2-Pdet}
\end{figure}

\textit{(i) Example 1 }: $L=10$. Following Eq.~\eqref{eq:mat_pdet}, we compute $P_{\det}$ for system-size $L=10$ by varying $\ket{s}$ and $\ket{d_2}$ as shown in Fig.~\ref{fig:NN-rank-2-Pdet}, consistent with the above analysis of dark and bright states. For $d_2=5$, we have $m_k=k\forall k$, i.e., \textit{Case II} is satisfied. Moreover, both $s=0$ and $5$ are return problems having $P_{\det}=1$ which follows from Corollary~\ref{cor:return}. Except for $d_2\neq 5$, $d_2 \ne \frac{5 m_k}{k}$ $\forall d_2,k$ which correspond to \textit{Case I} and consequently according to Theorem~\ref{th1}, $P_{\det}=1$ independent of 
the initial state $\ket{s}$. In both cases, either of \textit{Case I} or \textit{Case II} is satisfied for all degenerate energy eigenvalues $k$. But this is not always true as one increases the system size. Note that from the framework of dark and bright energy subspace, when no dark states exist, it immediately implies that \(P_{det}\) is unity  which does not require  any numerical computation. 

\textit{(ii) Example 2}: $L=20$, \textit{dependency of \(P_{\det}\) on dark states.} To show the dependence of $P_{\det}$ on number of dark states of a system,  we take $L=20$. According to Eq.~\eqref{rank-2_arb}, we fix $d_1=0$ and vary $d_2=1,2,\ldots,19$, i.e. $\mathbb{D}=\ketbra{0} + \ketbra{d_2}$. For each doubly degenerate energy level $\{E_k|k=1,2,\ldots,9\}$ we have
\begin{eqnarray}
    \nonumber \mathcal{A}_k &=& \begin{bmatrix}
    \braket{0}{\epsilon_k} & \braket{0}{\epsilon_{20-k}} \\ 
    \braket{d_2}{\epsilon_k} & \braket{d_2}{\epsilon_{20-k}} \\ 
    \end{bmatrix} \\ \nonumber 
    &=& \frac{1}{20}\begin{bmatrix}
    1 & 1 \\
    \exp(\frac{i2\pi k d_2}{20}) & \exp(\frac{-i2\pi k d_2}{20}) \\
    \end{bmatrix}\\
    \implies  REF(\mathcal{A}_k)&=&\begin{bmatrix}
                1 & 1 \\
                0 & 2 i\sin\left(\frac{\pi k d_2}{10}\right) \\  
            \end{bmatrix}.
\end{eqnarray}
The possibilities of $\rank(\mathcal{A}_k)$ can be
\begin{equation}
    \rank(\mathcal{A}_k) = \begin{cases}
        2 & \text{ if } \frac{kd_2}{10} \ne m_k  \\
        1 & \text{ if } \frac{kd_2}{10} = m_k,
    \end{cases}
\end{equation}
for some $ m_k \in \mathbb{Z}^{+}$.
\begin{figure}
    \centering
    \includegraphics[width=\linewidth]{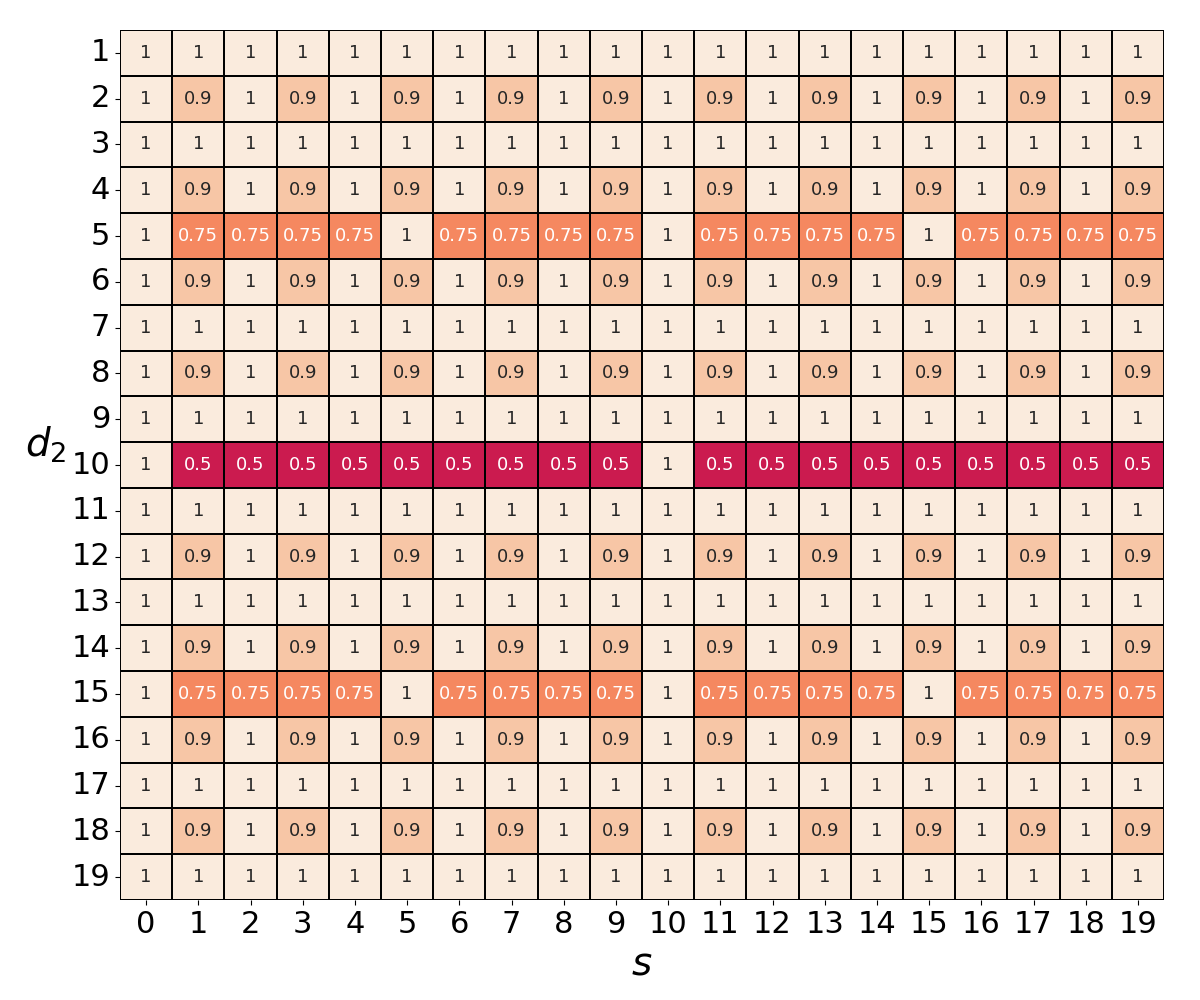}
    \caption{(Color online.) $P_{\det}$ for a cyclic graph with $L=20$. The two-dimensional subspace detector is $\mathbb{D}=\ketbra{0} + \ketbra{d_2}$, $d_2=1,2,\cdots,19$ similar to the case in Fig. \ref{fig:NN-rank-2-Pdet}. For a fixed initial state, i.e., for a fixed column, a range of values for $P_{\det}$ is observed which depends on the choice of the site,  $d_2$ . Corresponding to a particular value of $d_2$, the dark states present in the system change (see Table \ref{tab:pdet_l20}), thus influencing total detection probability $P_{\det}$. The presence of a fewer dark states corresponds to higher $P_{\det}$ for the subspace. For a similar reason as in Fig. \ref{fig:NN-rank-2-Pdet}, the row with $d_2=\frac{L}{2}=10$ has $P_{\det}=0.5$ for all initial states except $\ket{0}$ and $\ket{10}$ }. 
    \label{fig:L=20}
\end{figure}
We calculate $P_{\det}$ as shown in Fig.~\ref{fig:L=20}. Let us explain the probabilities serially:
\begin{enumerate}
    \item When value of $d_2$ is odd (except $5,15$), \textit{Case I} is satisfied $\forall k$, i.e., $P_{\det}=1$ independent of initial state $\ket s$ as there exist no dark energy states in the system.

    \item For even values of $d_2$ (except $10$), only the $E_5$ energy level has a dark state. 

    \item For $d_2=5,15$, the energy levels corresponding to $k=2,4,6,8$ consist of a single dark state each.

    \item Finally for $d_2=10$, each of  the energy levels of $k=2,4,5,6,8$ have a single dark state each.
\end{enumerate}
\textbf{Decrease in \(P_{\det}\) with increasing number of dark states.} To illustrate that the total detection probability  decreases monotonically with the increase of dark states, we fix the initial state $\ket{s}$ at site, $s=1$ and vary $d_2$ from Eq.~\eqref{rank-2-D}. See Table~\ref{tab:pdet_l20} for $P_{\det}$ by varying $\ket{d_2}$.
\begin{table}[h]
    \centering
    \begin{tabular}{|c|c|c|}
         \hline
        $d_2$& Total number of dark states & $P_{\det}$ \\
        \hline
        $1,3,7,9,11,13,17,19$& $0$& $1$\\
        \hline
        $2,4,6,8,12,14,16,18$& $1$& $0.9$\\
        \hline
        $5,15$& $4$& $0.75$\\
        \hline
        $10$& $5$& $0.5$ \\
        \hline 
    \end{tabular}
    \caption{Value of $P_{\det}$ and total number of dark states of the periodic NN system given in Eq.~\eqref{H for NN} with $L=20$ and $\mathbb{D}=\ketbra{0}{0}+\ketbra{d_2}{d_2}$, by varying $d_2$.}
    \label{tab:pdet_l20}
\end{table}
Despite the existence of dark states, the initial states that give $P_{\det}$ to be unity (e.g., $d_2=6$ and $s=4$) are due to Proposition~\ref{prop:phin_bright}. 


\subsubsection{Detection of three-dimensional subspace through rank-3 detector ($\tilde r = 3$)}

In case of rank-$3$ detector, we fix $d_1=0$, i.e., the detector has the form 
\begin{eqnarray}
    \mathbb D = \ketbra{0} + \ketbra{d_2} + \ketbra{d_3},
\end{eqnarray}
and vary $d_2=1, 2, \hdots L-1$, and $d_3 =1, 2, \hdots L/2-1$ with $d_2 \ne d_3$. Notice that we only vary $d_3$ from $1$ to $L/2 -1$ as the system is symmetric about the axis passing through $0$ and $L/2$. The matrix $\mathcal{A}_k$ and its row echelon form can be written as 
\begin{eqnarray}
\mathcal{A}_k&=& \begin{bmatrix}
    \braket{0}{\epsilon_k} & \braket{0}{\epsilon_{L-k}} \\
    \braket{d_2}{\epsilon_k} & \braket{d_2}{\epsilon_{L-k}} \\
    \braket{d_3}{\epsilon_k} & \braket{d_3}{\epsilon_{L-k}} 
\end{bmatrix} \\
&=& \frac{1}{L^{3/2}}  \begin{bmatrix}
                    1 & 1 \\
                    \exp(\frac{i 2 \pi k d_2}{L}) & \exp(\frac{-i 2 \pi k d_2}{L}) \\
                    \exp(\frac{i 2 \pi k d_3}{L}) & \exp(\frac{-i 2 \pi k d_3}{L}) \\ 
                \end{bmatrix},
\end{eqnarray}
and
\begin{eqnarray}
REF(\mathcal{A}_k)&=& \begin{bmatrix}
                1 & 1 \\
                0 & 2i\sin(\frac{2 \pi k d_2}{L}) \\
                0 & 2i\sin(\frac{2 \pi k d_3}{L}) \\ 
          \end{bmatrix},                
\end{eqnarray}
respectively. 

{\it Condition I. }We can see that $\rank(\mathcal{A}_k)=2$ as long as $ d_2\neq \frac{m_{k} L}{2k}$ or  $d_3\neq n d_2$ for some integer $m_{k}$ and $n\geq 2$ which shows that there is no dark state in degenerate subspace $E_k$.

{\it Condition II. }On the other hand, when $ d_2= \frac{m_{k} L}{2k}$ and  $d_3=n d_2$ are simultaneously satisfied, $\rank(\mathcal{A}_k)=1$ and the corresponding dark and bright states are same as in Eqs.~\eqref{eq:dark_r2} and \eqref{eq:bright_r2} respectively.

For $L=10$, the particle can always be detected deterministically (see Appendix~\ref{sec:app-nn_r3}) independent of $d_2$, $d_3$ and $s$ when $\tilde r=3$. In this case, only possible value of $m_k$ is $k$, i.e., $d_2=5$, but $d_3=5n$ does not exist inside the system. Therefore, {\it Condition I} is satisfied simultaneously for all degenerate energy levels, leading to the nonexistence of dark states in the system. This observation implies that $P_{\det}=1$ in accordance with Theorem~\ref{th1}. This is due to the fact that for any given subspace of rank $\tilde r=3$ in NN hopping model with $L=10$, we have $\tilde r>\max_k g_k=2$. Thus, we explain the deterministic nature of subspace detection in $L=10$ as shown in Fig.~\ref{fig:NN-rank-3-Pdet}.


For the NN system, the maximum degeneracy in the system is two-fold, and  each degenerate subspace consists of  at most one dark state. To illustrate the method given in Sec.~\ref{sec:br_vs_dark_energy}, i.e., to construct more than one dark states for subspace detection,  we incorporate the next nearest-neighbor hopping to the system along with NN in the next subsection. 


\subsection{NN and next nearest-neighbor (NNN) interacting system}


The system with NN and NNN hopping is governed by the Hamiltonian,
\begin{equation}
    H_1 = -\gamma \sum_{i=0}^{L-1} \ketbra{i+1}{i}  + \ketbra{i+2}{i} + h.c.,
    \label{eq:H-for-NN-and-NNN}
\end{equation}
where we take the system size to be $L=10n$ and $n\in \mathbb{Z}^+$. The eigenvalues and eigenvectors of $H_1$ are given by
\begin{equation}
    \begin{split}
        \epsilon_p &= -2\gamma \Big[ \cos(\frac{2\pi p}{L}) + \cos(\frac{4\pi p}{L}) \Big],  \\
        \ket{\epsilon_p} &= \frac{1}{\sqrt{L}} \sum_{j=0}^{L-1} \exp(\frac{i2\pi p j}{L}) \ket{j}.
    \end{split}
    \label{eq:NN-and-NNN-evals-and-evecs}
\end{equation}
We relabel the energy spectrum as follows:
\begin{itemize}
    \item \textit{Non-degenerate}: $\ket{E_0}=\ket{\epsilon_0}$ and $\ket{E_{L/2}}=\ket{\epsilon_{L/2}}$ are non-degenerate.

    \item \textit{Four-fold degenerate}: The set of  eigenstates corresponding to the single four-fold degenerate energy level can be represented as $\big\{\ket{E_{L/5,1}}=\ket{\epsilon_{L/5}},\ket{E_{L/5,2}}=\ket{\epsilon_{2L/5}},\ket{E_{L/5,3}}=\ket{\epsilon_{3L/5}},\ket{E_{L/5,1}}=\ket{\epsilon_{4L/5}}\big\}$.

    \item \textit{Two-fold degenerate}: $\frac{L-6}{2}$ number of two-fold degenerate energy levels can be written as $\big\{\ket{E_{k,1}}=\ket{\epsilon_k}\text{ and }\ket{E_{k,2}}=\ket{\epsilon_{L-k}}|k\neq L/5,2L/5\big\}_{k=1}^{(L/2)-1}$.
\end{itemize}
Let us analyze the total detection probability when a rank-$2$ detector is used to detect the particle. 
\subsubsection{Detecting the walker in the two-dimensional subspace ($\tilde r =2$)}
In the case of finding a particle in two-dimensional subspace, the projector is constructed as
\begin{equation}
    D = \ket{0}\bra{0} + \ket{d_2}\bra{d_2} \text{  with } d_2\ne 0,
    \label{eq:D-rank-2-NNN}
\end{equation}
Similar to the system with NN hopping, $\ket{E_0}$ and $\ket{E_{L/2}}$ are completely bright.

\textit{Four-fold degenerate subspace.} In the case of four-fold degenerate energy levels, we have
\begin{widetext}
    \begin{eqnarray}
    \mathcal{A}_{L/5} &=& 
    \begin{bmatrix}
        \braket{0}{E_{L/5,1}} & \braket{0}{E_{L/5,2}} & \braket{0}{E_{L/5,3}} & \braket{0}{E_{L/5,4}} \\
        \braket{d_2}{E_{L/5,1}} & \braket{d_2}{E_{L/5,2}} & \braket{d_2}{E_{L/5,3}} & \braket{d_2}{E_{L/5,4}} \\
      \end{bmatrix}, \nonumber \\
&=& \frac{\exp\left(\frac{i 2\pi d}{5}\right)}{L^2}                                 \begin{bmatrix}
            1 & 1 & 1 & 1 \\
            1 & \exp\left(\frac{i 2\pi d_2}{5}\right) & \exp\left(\frac{i 3\pi d_2}{5}\right) & \exp\left(\frac{i 4\pi d_2}{5}\right) \\ 
      \end{bmatrix}, \\
\implies REF(\mathcal{A}_{L/5})&=& 
    \begin{bmatrix}
        1 & 1 & 1 & 1 \\
        0 & \exp\left(\frac{i 2\pi d_2}{5}\right) -1 & \exp\left(\frac{i 3\pi d_2}{5}\right) - 1 & \exp\left(\frac{i 4\pi d_2}{5}\right) -1  \\ 
    \end{bmatrix}.
\end{eqnarray}
\end{widetext}
From $REF(\mathcal{A}_{L/5})$, it is clear that $\rank(\mathcal{A}_{L/5})=2$, when $\exp\left(\frac{i 2\pi d_2}{5}\right)=1, \exp\left(\frac{i 3\pi d_2}{5}\right)=1, $ and $\exp\left(\frac{i 4\pi d_2}{5}\right)=1$ are not simultaneously satisfied, i.e., $d_2\neq10m$ $\forall m\in\mathbb Z^+$. Following the procedure in Sec.~\ref{sec:br_vs_dark_energy}, two dark states in the case of the four-fold degenerate subspace can be calculated as
\begin{widetext}
    \begin{eqnarray}
    \nonumber \ket{\zeta_{L/5,1}} &=& N_1\Big[e^{\left(i\pi d\right)}\ket{E_{L/5,1}} - 2\cos\left(\frac{\pi d_2}{5}\right) e^{\left(\frac{i 4\pi d_2}{5}\right)} \ket{E_{L/5,2}} \nonumber + e^{\left(\frac{i 3\pi d_2}{5}\right)} \ket{E_{L/5,3}} \Big],\\ \nonumber 
    \ket{\zeta_{L/5,2}} &=& N_2 \Bigg[ 2i\cos^{2}\left(\frac{\pi d_2}{5}\right)\sin\left({\frac{\pi d_2}{5}}\right) \ket{E_{L/5,1}} \nonumber - \frac{i}{2}\Bigg(e^{\left(\frac{i8\pi d_2}{5}\right)}\sin\left(\frac{3\pi d_2}{5}\right) + e^{\left(\frac{i12\pi d_2}{5}\right)}\sin\left(\frac{\pi d_2}{5}\right) \Bigg) \ket{E_{L/5,2}} \nonumber \\ &+&  \frac{e^{\left(\frac{i6\pi}{5}\right)}}{2} \Big[2\left(e^{\left(\frac{i6\pi}{5}\right)} -1\right) \cos\left(\frac{\pi d_2}{5} \right) + i \sin\left(\frac{3\pi d_2}{5}\right)\Big] \ket{E_{L/5,3}} - i e^{\left(\frac{i 8\pi d_2}{5}\right)} \sin\left(\frac{3\pi d_2}{5} \right) \ket{E_{L/5,4}} \Bigg],
\end{eqnarray}
\end{widetext}
with the normalization constants, $N_1 = \frac{1}{\sqrt{2[1 +\cos^{2}\left(\frac{\pi d_2}{5}\right)]}}$ and $N_2 = \sqrt{\frac{2}{\left[20 + 27\cos\left(\frac{2\pi d_2}{5}\right) + 6\cos\left(\frac{4\pi d_2}{5}\right) + 7\cos\left(\frac{6\pi d_2}{5}\right) \right] \sin^{2}\left( \frac{\pi d_2}{5}\right)}}$. By Gram-Schmidt procedure, the bright states can be written as 
\begin{widetext}
\begin{eqnarray}
    \nonumber \ket{\eta_{L/5,1}}  &=& \frac{\mathbb{E}_{L/5}\ket{0}}{\sqrt{\bra{0}\mathbb{E}_{L/5}\ket{0}}}=\frac{1}{2}\left(\ket{E_{L/5,1}} + \ket{E_{L/5,2}} +  \ket{E_{L/5,3}} + \ket{E_{L/5,4}}\right),\\ \nonumber
    \ket{\eta_{L/5,2}}  &=&  \frac{\Big( \mathbb{E}_{L/5}\ket{d_2} - \bra{\eta_{L/5,1}}\mathbb{E}_{L/5}\ket{d_2} \ket{\eta_{L/5,1}} \Big)}{\sqrt{\bra{d_2}\mathbb{E}_{L/5}\ket{d_2} - \bra{d_2}\mathbb{E}_{L/5}\ket{\eta_{L/5,1}}\bra{\eta_{L/5,1}}\mathbb{E}_{L/5}\ket{d_2}}} \\ \nonumber &=&  B\Big[\left(e^{\frac{i4\pi d_2}{5}} + e^{\frac{i6\pi d_2}{5}} + e^{\frac{i8\pi d_2}{5}} \right) \ket{E_{L/5,1}}  + \left(e^{\frac{i2\pi d_2}{5}} + e^{\frac{i6\pi d_2}{5}} + e^{\frac{i8\pi d_2}{5}} \right) \ket{E_{L/5,2}} \nonumber \\ &+& \left(e^{\frac{i2\pi d_2}{5}} + e^{\frac{i4\pi d_2}{5}} + e^{\frac{i8\pi d}{5}} \right)  \ket{E_{L/5,3}} \\  &+& \left(e^{\frac{i2\pi d_2}{5}} + e^{\frac{i4\pi d_2}{5}} + e^{\frac{i6\pi d_2}{5}} \right)\ket{E_{L/5,4}} \Big], \\
    B &=& \frac{1}{\sqrt{4(1 - 1\left[\cos\left(2\pi d_2/5 \right) + \cos\left(4\pi d_2/5 \right)\right]^{2})}}.
\end{eqnarray}
\end{widetext}
On the other hand, for $d_2=10m$, $\rank({\mathcal{A}_k})= 1$, hence the three dark states and a single bright state are given by 
\begin{align}
    \nonumber \ket{\zeta_{L/5,1}} &=  \frac{1}{\sqrt{2}}\left[\ket{E_{L/5,2}} - \ket{E_{L/5,1}}\right],  \\ \nonumber
    \ket{\zeta_{L/5,2}} &=  \frac{1}{\sqrt{6}}\left[2\ket{E_{L/5,3}} - \ket{E_{L/5,2}} - \ket{E_{L/5,1}} \right],  \\ \nonumber
    \ket{\zeta_{L/5,3}} &= \frac{1}{2\sqrt{3}}\big[3\ket{E_{L/5,4}} - \ket{E_{L/5,3}} - \ket{E_{L/5,2}}\\ \&- \ket{E_{L/5,1}} \big], \label{eq:NNN-zeta-rank-2-2}
\end{align}
\begin{align}    
    \nonumber \ket{\eta_{L/5}} &= \frac{1}{2} \big[ \ket{E_{L/5,1}} + \ket{E_{L/5,2}} + \ket{E_{L/5,3}} \\ &+ \ket{E_{L/5,4}} \big].
    \label{eq:NNN-eta-rank-2-2}
\end{align}
\textit{Two-fold degenerate subspace. } In this scenario, the condition for the dark state gives
\begin{eqnarray}
    \mathcal{A}_k &=& 
    \begin{bmatrix}
        \braket{0}{E_{k,1}} & \braket{0}{E_{k,2}} \\
        \braket{d_2}{E_{k,1}} & \braket{d_2}{E_{k,2}}
      \end{bmatrix} \nonumber \\
  &=& \frac{1}{L}
  \begin{bmatrix}
        1 & 1 \\
        \exp\left(\frac{i 2\pi k d_2}{L}\right) & \exp\left(\frac{-i 2\pi k d_2}{L}\right) \\ 
    \end{bmatrix},  \\
    REF(\mathcal{A}_k) &=& 
    \begin{bmatrix}
        1 & 1 \\
        0 & 2 i \sin\left(\frac{ 2\pi k d_2}{L}\right) \\ 
    \end{bmatrix}.
\end{eqnarray}
When $d=\frac{m_k L}{2k}$ for some $m_k\in \mathbb Z^+$, we have $\rank(\mathcal{A}_k)=1$ and the corresponding dark and bright states are given by
\begin{eqnarray}
    \ket{\zeta_k} =  \frac{\ket{E_{k,1}} - \ket{E_{k,2}}}{\sqrt{2}}, \\
    \ket{\eta_k} =  \frac{\ket{E_{k,1}} + \ket{E_{k,2}}}{\sqrt{2}}.
\end{eqnarray}
Otherwise, the two-fold energy level has no dark states.

We take the system-size, $L=10$, $\ket{\phi(0)}=\ket s$ and using this analysis, we compute the total probability of detection for $d_2\neq 5$ as
\begin{eqnarray}
    P_{\det}(s) &=& \bra{s}\mathbb{P}_{\mathcal{H}_{\eta}}\ket{s} \nonumber \\
    &=& \frac{6}{10} + \frac{1}{10}\Big[\cos(\frac{2 s\pi}{5}) + \cos(\frac{4 s\pi}{5})\Big]^2 \nonumber \\
    &+& \frac{B^{2}}{10} \Big[\frac{3}{4}\cos(\frac{2(s-d_2)\pi}{5}) + \frac{3}{4}\cos(\frac{4(s-d_2)\pi}{5}) \nonumber \\ 
    &-& \frac{1}{4}\cos(\frac{2(s+d_2)\pi}{5}) - \frac{1}{4}\cos(\frac{4(s+d_2)\pi}{5}) \nonumber \\
    &-&\frac{1}{4}\cos(\frac{2(s+2d_2)\pi}{5}) - \frac{1}{4}\cos(\frac{2(s-2d_2)\pi}{5}) \nonumber \\  
    &-& \frac{1}{4}\cos(\frac{2(2s+d_2)\pi}{5}) - \frac{1}{4}\cos(\frac{2(2s-d_2)\pi}{5}) \Big]^2 \nonumber \\ 
    &=& \begin{cases}
        1 & s = 0, d_2, 5, d_2+5 \\
        2/3 & otherwise,
    \end{cases}
    \label{eq:NNN-Pdet-expr-rank2-1}
\end{eqnarray}
and for $d_2=5$ as
\begin{eqnarray}
    P_{\det}(s) &=& \frac{2}{10} + \frac{2}{10}\cos^{2}\left(\frac{2\pi s}{10}\right) + \frac{2}{10}\cos^{2}\left(\frac{6\pi s}{10}\right) \nonumber \\ 
    &\text{}& + \frac{1}{10} \left[ \cos^{2}\left(\frac{4\pi s}{10} \right) + \cos^{2}\left(\frac{8\pi s}{10} \right)  \right] \nonumber \\
    &=& \begin{cases}
            1 & s = 0, 5 \\
            3/8 & otherwise.
        \end{cases}
    \label{eq:NNN-Pdet-expr-rank2-2}
\end{eqnarray}
Finally, $P_{\det}$, calculated using Eq.~\eqref{eq:mat_pdet} as shown in Fig.~\ref{fig:NNN-rank-2-Pdet} agrees with the one obtained via 
\begin{figure}[htbp]
    \centering
    \includegraphics[width=\linewidth]{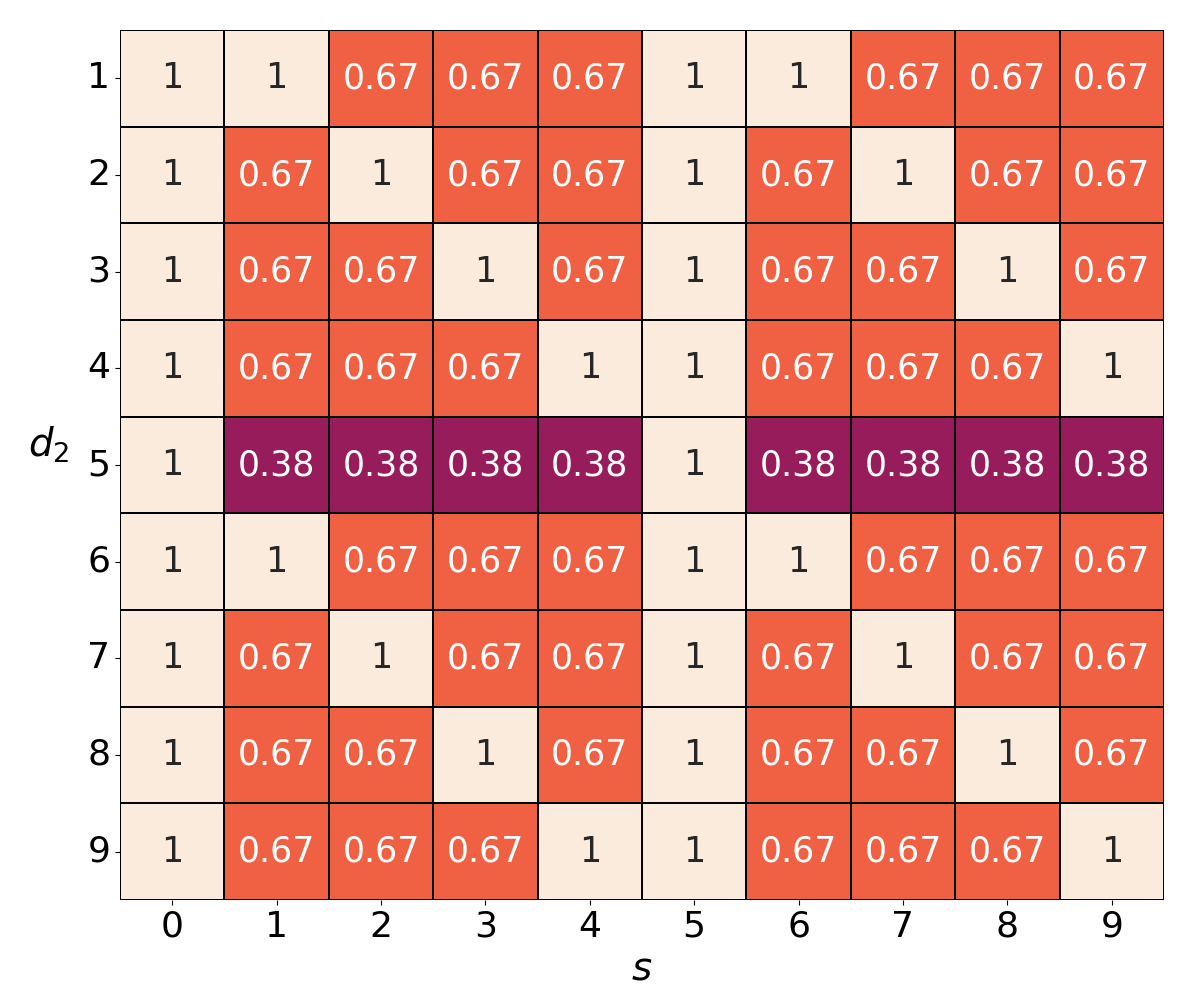}
    \caption{(Color online.) Two-dimensional subspace detection probability for a ring with additional next nearest-neighbor hopping ($L=10$) as given in Eq. (\ref{eq:H-for-NN-and-NNN}) and $\mathbb{D}=\ketbra{0} + \ketbra{d_2}$, $d_2=1,2,\cdots,9$. Notice that the values of $P_{det}$ are rounded up to two decimal places. Further, we notice that when $d_2=5$, $P_{\det}=0.375 \approx 0.38$ for all initial states except $\ket{0}$ and $\ket{5}$.}
    \label{fig:NNN-rank-2-Pdet}
\end{figure}
the approach with bright and dark states (see Eqs.~\eqref{eq:NNN-Pdet-expr-rank2-1} and \eqref{eq:NNN-Pdet-expr-rank2-2}).

\section{Average number of  measurements for subspace detection}
\label{sec:avnomeasure}

Throughout the paper, we have adhered to stroboscopic measurement protocol for detecting particle in a given subspace in quantum walk evolution. Along with the total detection probability $P_{\det}$ after an infinite number of measurement attempts, another quantity of interest is the time  required to detect the particle in the corresponding subspace, also known as the hitting time of the quantum walk \cite{Varbanov2008}. Intermediate free evolution time $\tau$ being constant, the average number of measurements, $\bar n$, calculated as $\bar{n}=\sum_{n=1}^{\infty}nF_n$, is required to detect the particle which means hitting time $\bar{\tau}=\tau \bar{n}$. Note, however from our previous discussions, $P_{\det}$ can be less than unity where average number of measurements required for detection can not be calculated by above mentioned formula as the detection is not guaranteed. Therefore, we redefine the average number of measurements conditioned on the fact that the particle is detected as 
\begin{eqnarray}
    \bar n = \frac{\sum_{n=1}^{\infty}nF_n}{P_{\det}}.
\end{eqnarray}
Furthermore, depending on the initial state of the quantum walk, the hitting time can again be classified into two distinct categories -
    \begin{enumerate}
        \item Average \textit{arrival time} for quantum walks starting from initial states orthogonal to the detection subspace.
        \item Average \textit{return time} for quantum walks starting with initial states within the detection subspace. 
    \end{enumerate}
Before presenting examples, let us first show that $\bar n$ can have a closed-form expression by which it can be calculated efficiently\ \cite{Varbanov2008,Kessler2021}.
\begin{proposition}
The average number of measurements required for the detection of the particle in a subspace conditioned on the fact that a particle is detected can be written in a closed-form as 
\begin{equation}
    \bar{n} = \frac{1}{P_{\det}}\sum_{i=1}^{r} \Tr\left[ \mathcal{L}^{-2} \bar{U} \bar{S} \bar{W} \bar{T_{i}^{*}} \right].
\end{equation}
\end{proposition}
\begin{proof}
    Using Eq.~\eqref{eq:fn} the average number of measurements $\bar{n}$ conditioned on successful detection can be written as 
\begin{eqnarray}
        \nonumber \bar{n} &=& \frac{1}{P_{\det}}\sum_{i=1}^{\infty} n F_n \\ \nonumber
        &=& \frac{1}{P_{\det}}\sum_{n=1}^{\infty} \sum_{i=1}^{\tilde{r}} n\Tr\left[ (\bar{U}\bar{C})^{n-1} \bar{U} \bar{S} \bar{W} \bar{T_{i}^{*}} \right]\\
        &=& \frac{1}{P_{\det}}\sum_{i=1}^{r} \Tr\left[ \mathcal{L}^{-2} \bar{U} \bar{S} \bar{W} \bar{T_{i}^{*}} \right].
\end{eqnarray}
\end{proof}
Let us now calculate $\bar n$ for different configurations in case of NN interacting periodic system as mentioned in Eq.~\eqref{H for NN} with system-size $L=10$. 
\begin{figure}
    \centering
    \includegraphics[width=\linewidth]{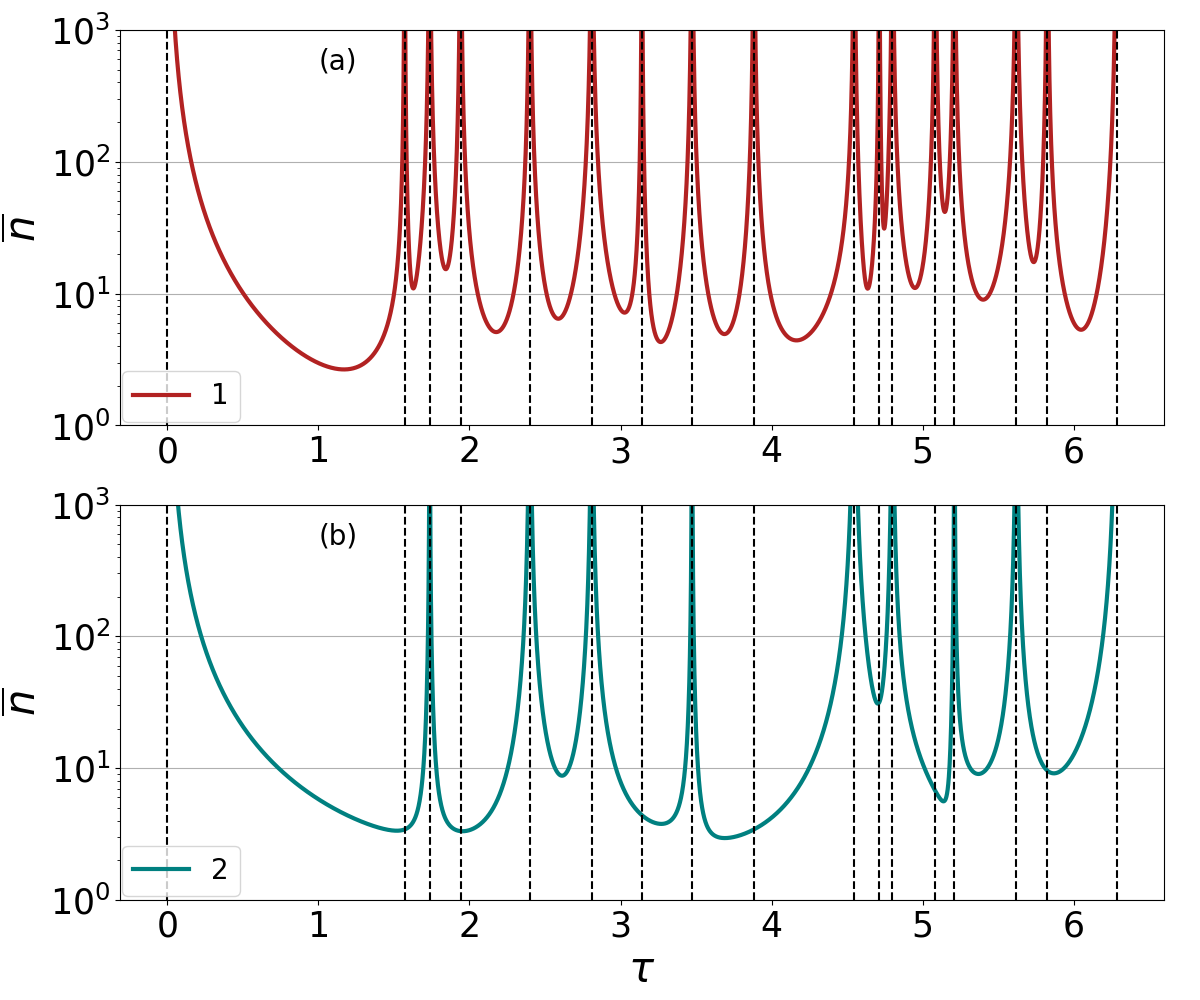}
    \caption{(Color online.)  $\bar n$ (ordinate) against free evolution time $\tau$ (abscissa) for system-size $L=10$ of a cyclic graph with NN hopping (Eq.~\eqref{H for NN}). Particle is initially at site $s=1$ (top) and $s=2$ (bottom) and the detector is $\mathbb{D}=\ketbra{0}{0}$. Comparing the plots, we demonstrate that $\bar n$ may be divergent only at $\tau_c$. All the axes are dimensionless.}
    \label{fig:rank1-10-nn}
\end{figure}
\subsection{Arrival time}
At first, for single site detection we take the detector to be fixed at $d=0$ and calculate $\bar n$ where the walk is taken to be initially at site $s=1$ or $s=2$ as shown in Fig.~\ref{fig:rank1-10-nn}. Notice that $\bar n$ diverges at certain times for both cases which is possible only when the critical condition \cite{Liu2020,Kessler2021} $|\epsilon_{p'} -\epsilon_p|\tau_c = 2 n \pi$ (with $n\in \mathbb{Z}^+$) satisfies for any pair of energy eigenvalues $\epsilon_{p'}$ and $\epsilon_p$ of $H$. When $\tau=\tau_c$s, the degeneracy of the evolution unitary $U=e^{-iHt}$ increases from the degeneracy of $H$. In Fig.~\ref{fig:rank1-10-nn}, the vertical dotted lines correspond to all the critical times $\tau_c$. However, divergence at $\tau=0$ signifies the Zeno limit \cite{thiel2020zeno}. Note that for $s=1$, at all $\tau_c$s, the value of $\bar n$ diverges. However, $\bar n$ does not diverge at all $\tau_c$s when $s=2$. This shows that although all the $\tau_c$s may not correspond to divergence seen in $\bar n$, it belongs to the set of $\tau_c$s. Note further that the methods discussed in Secs. \ref{sec:br_vs_dark_energy}  and \ref{sec:matrix_method} for computing \(P_{det}\) remains valid in all the points in which \(\bar{n}\) does not diverge. 

Let us demonstrate that in the case of higher rank measurements, $\bar n$ diverges for some values of $\tau$s, which is a subset of the set of critical time $\tau_c$ obtained in rank-1 measurements. 
\begin{figure}
        \centering
        \includegraphics[width=\linewidth]{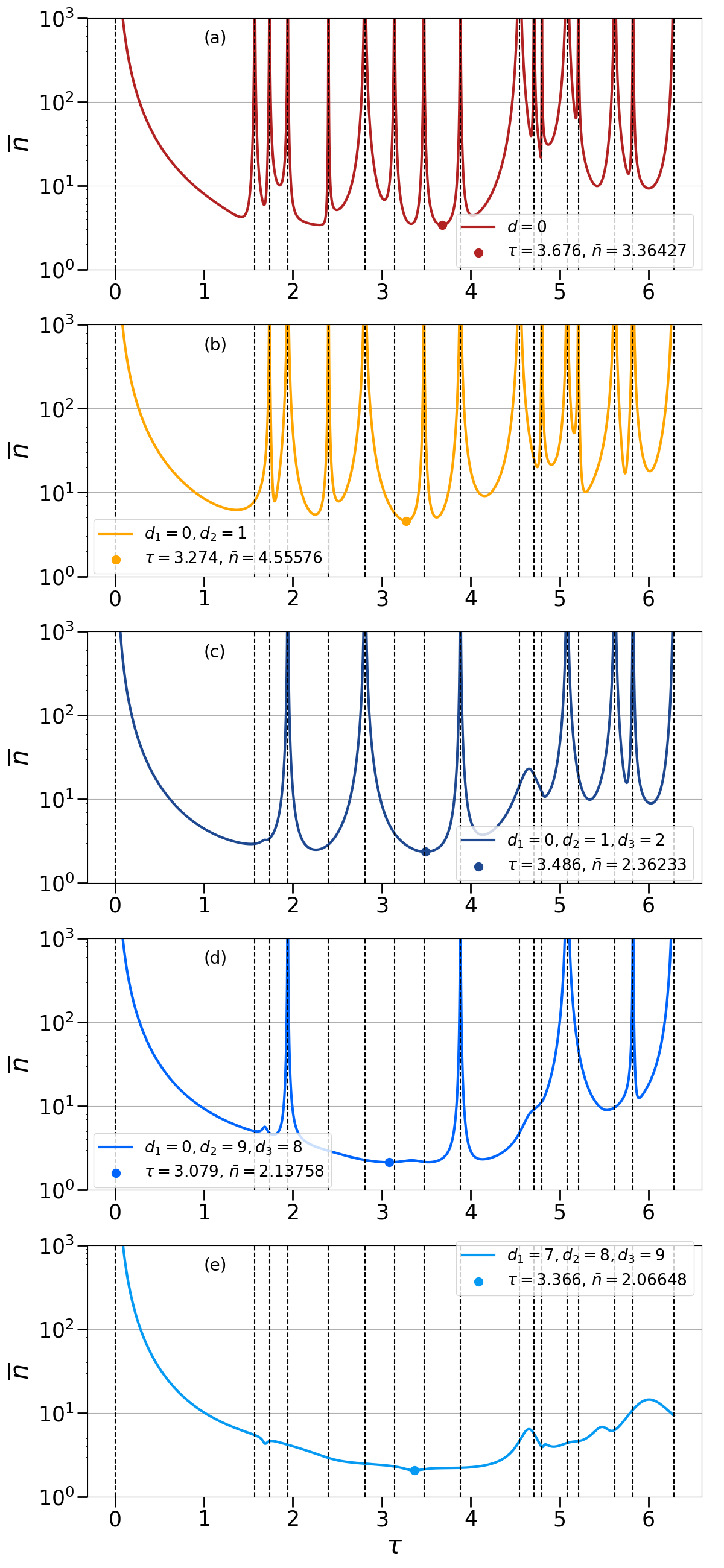}
        \caption{(Color online.) The average number of measurements conditioned on  the particle to be  detected, $\bar{n}$ (vertical axis) against $\tau$ (horizontal axis) for a quantum walk on the cyclic graph with $L=10$ in  Eq.~\eqref{H for NN} initialized at $\ket s =\ket{3}$. The subspaces of dimensions one ((a)), two ((b)), and three ((c)-(e)) are measured. Points in each plot represent  the minimum $\bar{n}$ and its corresponding \(\tau\), mentioned in the legend.  As mentioned in the text, the number of divergences in $\bar n$ can be reduced by suitably choosing rank-3 detector as seen from (c), (d) and (e). Interestingly, for diametrically opposite subspace detector, $\mathbb{D}=\ketbra{7}{7}+\ketbra{8}{8}+\ketbra{9}{9}$ ((e)), there is no divergence except at the Zeno limit. All the axis are dimensionless.} 
        \label{fig:nbar-10-s-3}
    \end{figure}
For illustration, we consider the initial state as $\ket{s}=\ket{3}$ and increase the rank of the detector subspace. The observations are listed below.
\begin{enumerate}
    \item Let us take $\mathbb{D} = \ketbra{0} + \ketbra{1}$ and $\mathbb{D}=\ketbra{0} + \ketbra{1} + \ketbra{2}$. In case of rank-1 detector, $\mathbb{D}=\ketbra{0}$, the number of divergences in $\bar{n}$ is $17$ (see Fig. \ref{fig:rank1-10-nn}) and they exist at all $\tau_c$ with $0\le \tau \le 2\pi$. On the other hand, the divergences in $\bar{n}$ reduce to $14$ and $8$ for the above rank-2 and rank-3 detectors respectively (see Fig. \ref{fig:nbar-10-s-3}).
    
    \item In case of rank-3 detector, instead of $d_2 = 1$ and $d_3 = 2$, if we take $d_2 = 9$ and $d_3 = 8$, the number of divergences in $\bar{n}$ is further reduced to $6$.
    \item Interestingly, when the detector subspace is exactly diametrically opposite to the initial state $\ket{3}$, i.e., $\mathbb{D}=\ketbra{7}+\ketbra{8}+\ketbra{9}$,  all divergences vanish except the Zeno limit. Note that the site $d_2=8$ is the diametrically opposite one for $\ket{s}=\ket{3}$. 
\end{enumerate}

Moreover, we compute the time ($\tau$) required for minimum $\bar n$, listed in Fig.~\ref{fig:nbar-10-s-3}. The foregoing analysis, in conjunction with the mitigation of divergences in the necessary number of measurements for detection, can be important from the perspective of the control of a quantum particle through a periodic measurement scheme, warranting further investigation.

\subsection{Return time}

In the case of a return problem, we consider the initial state to be localized in any single site of a given subspace. For a fixed initial state, $\ket s=\ket 0$, the return of the particle in $\mathbb{D}=\ketbra{0}{0}$ and $\mathbb{D}=\ketbra{0}{0}+\ketbra{1}{1}$ has constant $\bar n=6$ and $\bar n= 5$ respectively. Notably, for the subspace of rank greater than two, the average number of measurements required for the detection oscillates with $\tau$ as shown in Fig.~\ref{fig:return-nbar}, which requires critical analysis.

\begin{figure}
    \centering
    \includegraphics[width=\linewidth]{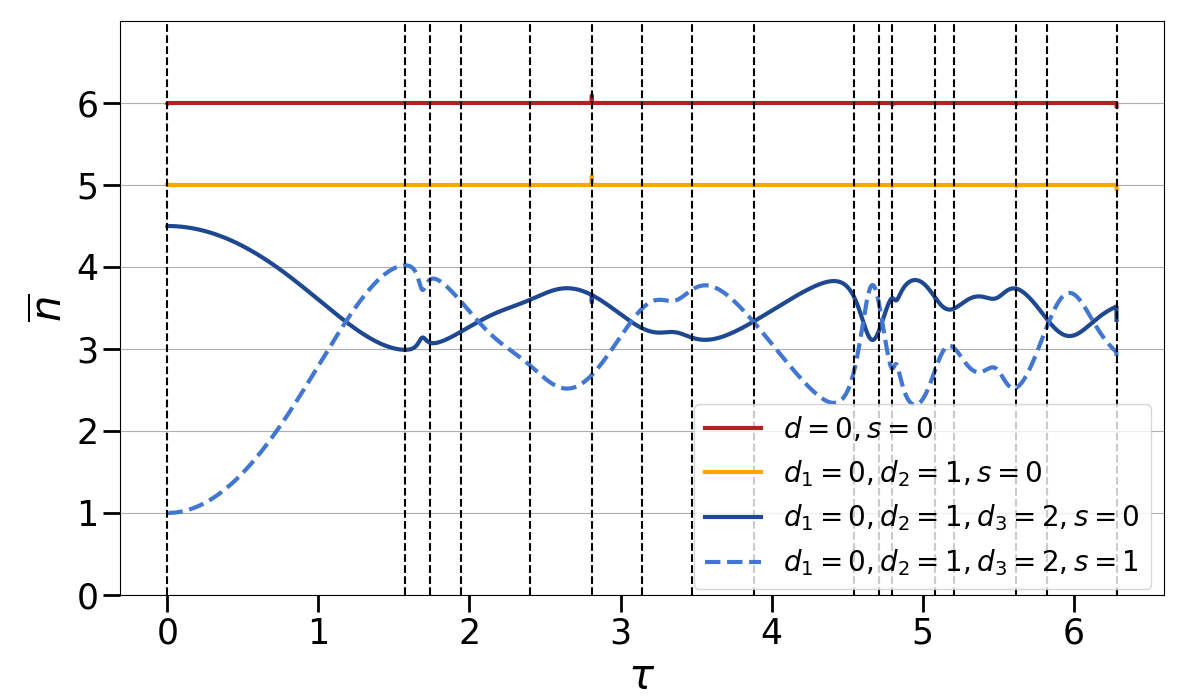}
    \caption{(Color online.) $\bar{n}$ (vertical axis) with $\tau$ (horizontal axis) for return problem in subspaces of dimension one, two, and three. The initial states are chosen such that $\bar{\tau}=\tau \bar{n}$ corresponds to the average return time of a quantum walk for a particular subspace. For initial state, $\ket s=\ket 0$, the value of $\bar n=6$ and $5$ for $\mathbb{D}=\ketbra{0}{0}$ and $\mathbb{D}=\ketbra{0}{0}+\ketbra{1}{1}$ respectively are independent $\tau$. On the other hand,  $\bar n$ attains oscillatory behavior with $\tau$ for $\mathbb{D}=\ketbra{0}{0}+\ketbra{1}{1}+\ketbra{2}{2}$ with initial states at the  site $s=0$ (solid line) or $1$ (dashed line). Note that $\bar n$ decreases with the increase in the rank of the detector. Both the axis are dimensionless. }
    \label{fig:return-nbar}
\end{figure}

\section{conclusion}
\label{sec:conclu}

    We  focused on the statistics of the particle's arrival in a given subspace undergoing measurement-induced quantum walk. In particular,  we obtained the first detection probability and the corresponding  total detection probability of a particle within  a subspace using the stroboscopic measurement protocol.
    
    We formulated  an alternative framework for detecting a particle in a  subspace  that utilizes the notion of dark and bright energy 
    eigenstates of a given Hamiltonian used for a quantum walk.
    Specifically, we employed the rank-nullity theorem to determine the number of dark and bright energy eigenstates for degenerate energy levels when we assert the detection of a particle in a defined subspace. 
    Based on the energy spectrum of the Hamiltonian, responsible for a quantum walk and its relationship with the detector state, we uncovered conditions independent of the choice of the initial state so that the particle performs a quantum walk to be certainly detected in a given subspace.  
    Note that such subspaces are perfect candidates for encoding target states in the context of quantum search problems. 
     In this configuration,  the possibility of the existence of multiple bright states corresponding to a degenerate energy level emerges which is not observed in the detection of the single site \cite{Barkai2020DarkState}.
    For a given rank of the subspace, we illustrated the monotonic decrease of detection probability with increase of dark states in the system depending on the position of the detectors. Furthermore,  by using alternative numerical method, we  computed the total detection probability  of continuous-time random walk in a cyclic graph with nearest-neighbor and next nearest-neighbor hopping which can also be confirmed using the approach with  bright and dark energy states. 
    
    In the context of arrival and return to the subspace, we found that the divergence observed in the average number of measurements conditioned on the successful detection  can be significantly reduced  when subspace detection of a particle is taken into account with the assistance of higher rank projectors. 
    The measuring process presented in this work is not identical with other divergence-removing techniques known in the literature and hence opens up an intriguing avenue for future exploration in case of quantum control. Our findings emphasize the significance of measurement strategies for conclusively detecting the particle and may suggest more study utilizing higher rank projectors in the context of quantum continuous-time random walks.

\acknowledgments
 
We acknowledge the support from the Interdisciplinary Cyber-Physical Systems (ICPS) program of the Department of Science and Technology (DST), India, Grant No.: DST/ICPS/QuST/Theme- 1/2019/23. We acknowledge the use of \href{https://github.com/titaschanda/QIClib}{QIClib} -- a modern C++ library for general-purpose quantum information processing and quantum computing (\url{https://titaschanda.github.io/QIClib}) and cluster computing facility at Harish-Chandra Research Institute.

\appendix
\section{Deriving exact formulas of first detection statistics in case of subspace detection}
\label{Appndx:Subspace}
Let us derive a closed-form expression for the first detection statistics in the case of subspace detection. For any operator $X$, $\bra{\Psi}X\ket{\Phi}$ in energy eigenbasis of the given Hamiltonian can be written as 
\begin{eqnarray}
        \nonumber \bra{\Psi}X\ket{\Phi} &=& \sum_{i,j}\langle\Psi\ketbra{\epsilon_i}{\epsilon_i}X\ketbra{\epsilon_j}{\epsilon_j}\Phi\rangle\\ \nonumber &=& \sum_{i,j} \Psi^{*}_{i}X_{ij}\Phi_{j} \\
        &=& \Tr\left[ X S W T^{*} \right]
        \label{eq:appndx-trace-picture}
\end{eqnarray}
where $\Psi_{i} = \braket{\epsilon_i}{\Psi}$, $\Phi_{i} = \braket{\epsilon_i}{\Phi}$.
The operators are as follows
\begin{eqnarray}
    \nonumber W_{ij} &=& 1, \\\nonumber
    S &=& \text{diag}(\Phi_1, \Phi_2, ..., \Phi_n),\\
    T &=& \text{diag}(\Psi_1, \Psi_2, ..., \Psi_n).
\end{eqnarray}
The unitary operator in the energy basis can be written as
 \begin{equation}
     U =\text{diag}(e^{-i \epsilon_1 \tau}, e^{-i \epsilon_2 \tau}, e^{-i \epsilon_3 \tau}, ..., e^{-i \epsilon_n \tau} ).
 \end{equation}   
Using Eq.~(\ref{Fn for subspace}), the first detection probability for $n=1$ is given by 
\begin{eqnarray}
    \nonumber F_1 &=& \bra{\phi(1)}\mathbb{D}\ket{\phi(1)} \\\nonumber
    &=& \sum_{i=1}^{r} \bra{\phi(0)} U^{\dag}\ketbra{d_i}U\ket{\phi(0)} \\\nonumber
    &=& \sum_{i=1}^{r} \Tr\left[U^{*} T_{i} W S^{*} \right] \Tr\left[U S W T_{i}^{*} \right] \\
    &=& \sum_{i=1}^{r} \Tr\left[U S W T_{i}^{*}\right]^{*} \Tr\left[U S W T_{i}^{*}\right],
\end{eqnarray}
where $S_{kk} = \braket{\epsilon_k}{\psi(0)}$ , $(T_{i})_{kk}=\braket{\epsilon_k}{d_i}$. Going from third to fourth line we use the cyclic property of trace along with the fact that $U, S, T_i$ commute as they are diagonal matrices.
\begin{widetext}
    Similarly, for $n=2$, we have
    \begin{eqnarray}
        F_2 &=& \bra{\phi(2)}\mathbb D\ket{\phi(2)} \nonumber \\
            &=& \sum_{i=1}^{\tilde r} \bra{\phi(0)} U^{\dag}(\mathbb{I}-\mathbb D)U^{\dag}\ket{d_i}\bra{d_i}U(\mathbb{I}-\mathbb D)U\ket{\phi(0)} \nonumber \\
            &=& \sum_{i=1}^{\tilde r} \Tr\left[ U\left(\mathbb{I} - \sum_{j=1}^{\tilde r}T_{j}^{*}WT_{j}\right) U S W T_{i}^{*} \right]^{*} \Tr\left[ U\left(\mathbb{I} - \sum_{j=1}^{\tilde r}T_{j}^{*}WT_{j}\right) U S W T_{i}^{*} \right] \nonumber \\
            &=& \sum_{i=1}^{\tilde r} \Tr\left[ U C U S W T_{i}^{*} \right]^{*} \Tr\left[ U C U S W T_{i}^{*} \right],
    \end{eqnarray}
\end{widetext}
where $C=\mathbb{I} - \sum_{j=1}^{\tilde r}T_{j}^{*}WT_{j}$. Finally, for the $n$th round, $F_n$ is calculated as 
\begin{equation}
    \begin{split}
        F_n &= \bra{\phi(n)}\mathbb D\ket{\phi(n)} \\ 
        &= \sum_{i=1}^{\tilde r}\Tr\left[ (U C)^{n-1} U S W T_{i}^{*} \right]^{*} \Tr\left[ (U C)^{n-1} U S W T_{i}^{*} \right]. \\
    \end{split}
    \label{eq:derived-Fn-appndx}
\end{equation}

Using the property of trace, given as 
\begin{equation}
    \begin{split}
        \Tr[A \otimes B] &= \Tr[A] \Tr[B]  \\
        \Tr[A^{*} \otimes A] &= \Tr[A^{*}] \Tr[A] = \Tr[A]^{*}\Tr[A],
    \end{split}
\end{equation}
and defining $\bar{O} = O^{*} \otimes O $, for any operator $O$ we can rewrite Eq.~(\ref{eq:derived-Fn-appndx}) as
\begin{equation}
    F_n = \sum_{i=1}^{\tilde r}\Tr\left[ (\bar{U} \bar{C})^{n-1} \bar{U} \bar{S} \bar{W} \bar{T_{i}^{*}} \right].
\end{equation}

Therefore, we can calculate $P_{\det}$ as  
\begin{equation}
    \begin{split}
        P_{\det} &= \sum_{n=1}^{\infty} F_n \\
        &= \sum_{n=1}^{\infty} \sum_{i=1}^{\tilde r} \Tr\left[ (\bar{U}\bar{C})^{n-1} \bar{U} \bar{S} \bar{W} \bar{T_{i}^{*}} \right] \\ 
        &= \sum_{i=1}^{\tilde r} \Tr\left[ \mathcal{L}^{-1} \bar{U} \bar{S} \bar{W} \bar{T_{i}^{*}} \right],
    \end{split}
\end{equation}
where we perform summation of the infinite geometric series and define $\mathcal{L} = \mathbb{I} - \bar{U}\bar{C}$.

\section{$P_{\det}$ for NN hopping system with rank-$3$ detector}
\label{sec:app-nn_r3}
The detector has the following configuration:
\begin{eqnarray}
    \mathbb D=\ketbra{0}{0}+ \ketbra{d_2} + \ketbra{d_3},
    \label{eq:det_r3}
\end{eqnarray}
where $d_2=1, 2, \hdots L$, and $d_3 =1, 2, \hdots L/2-1$ with $d_2 \ne d_3$. For $L=10$, we calculate $P_{\det}$ as shown in Fig.~\ref{fig:NN-rank-3-Pdet}.
\begin{figure}[h]
    \centering
    \includegraphics[width=0.9\linewidth]{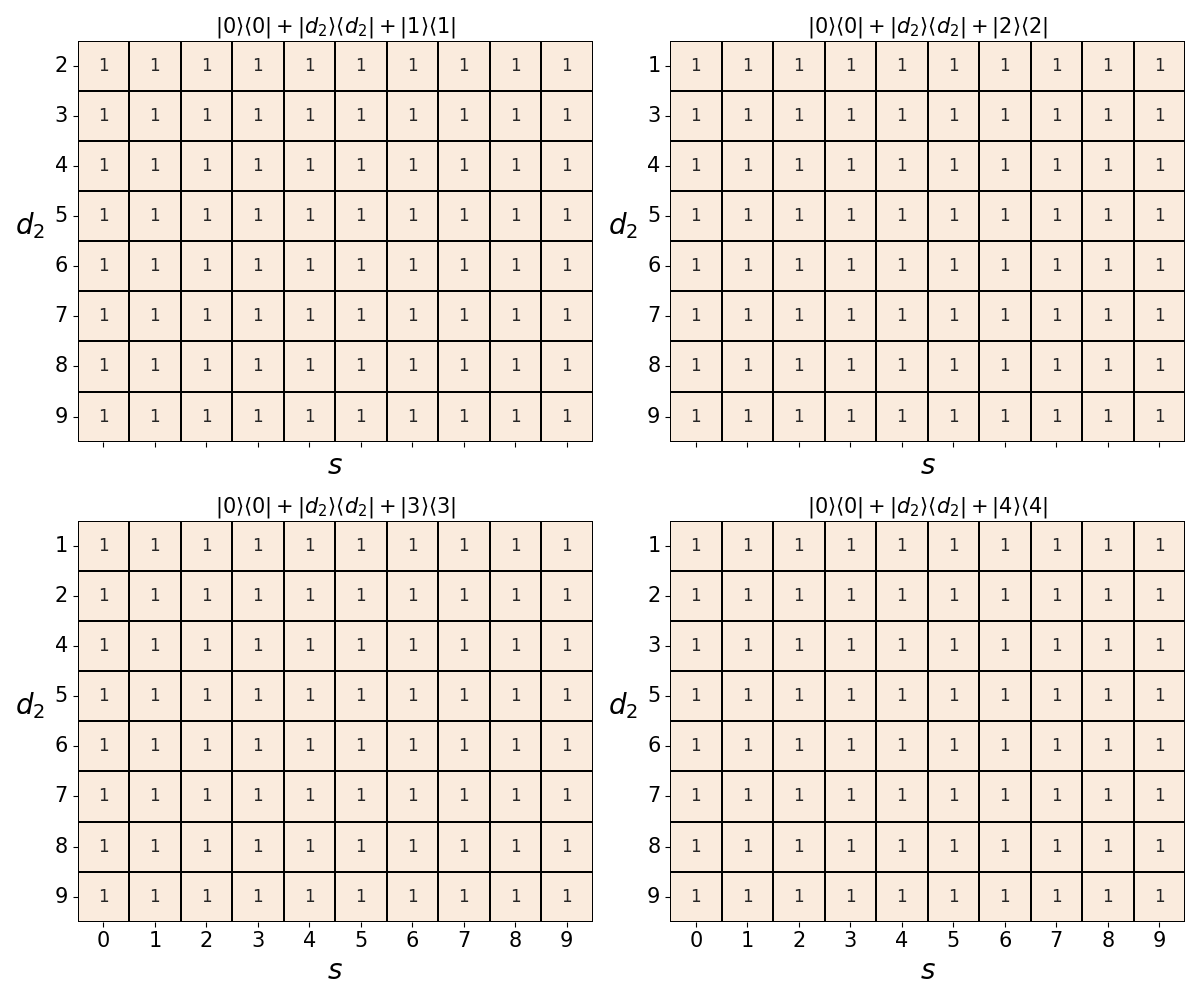}
    \caption{(Color online.) Three-dimensional subspace detection probability for the system-size $(L=10)$ as in Eq.~(\ref{H for NN}) and the detector given as Eq.~(\ref{eq:det_r3}). For all rank-3 detectors, the particle is certainly detected due to the absence of dark energy states as stated in Theorem \ref{th1}. }
    \label{fig:NN-rank-3-Pdet}
\end{figure}

\newpage

\bibliographystyle{apsrev4-1}
\bibliography{bib.bib}

\end{document}